\pdfoutput=1
\RequirePackage{ifpdf}
\ifpdf 
\documentclass[pdftex]{sigma}
\else
\documentclass{sigma}
\fi

\numberwithin{equation}{section}

\newtheorem{Theorem}{Theorem}[section]
\newtheorem{Corollary}[Theorem]{Corollary}
\newtheorem{Lemma}[Theorem]{Lemma}
\newtheorem{Proposition}[Theorem]{Proposition}
 {\theoremstyle{definition}
\newtheorem{Remark}[Theorem]{Remark} }

\begin{document}

\allowdisplaybreaks

\newcommand{\arXivNumber}{1609.02525}

\renewcommand{\thefootnote}{}

\renewcommand{\PaperNumber}{011}

\FirstPageHeading

\ShortArticleName{Series Solutions of the Non-Stationary Heun Equation}

\ArticleName{Series Solutions of the Non-Stationary Heun Equation\footnote{This paper is a~contribution to the Special Issue on Elliptic Hypergeometric Functions and Their Applications. The full collection is available at \href{https://www.emis.de/journals/SIGMA/EHF2017.html}{https://www.emis.de/journals/SIGMA/EHF2017.html}}}

\Author{Farrokh ATAI~$^{\dag\ddag}$ and Edwin LANGMANN~$^\dag$}

\AuthorNameForHeading{F.~Atai and E.~Langmann}

\Address{$^\dag$~Department of Physics, KTH Royal Institute of Technology, SE-10691 Stockholm, Sweden}
\EmailD{\href{mailto:langmann@kth.se}{langmann@kth.se}}
\Address{$^\ddag$~Department of Mathematics, Kobe University, Rokko, Kobe 657-8501, Japan}

\EmailD{\href{mailto:farrokh@math.kobe-u.ac.jp}{farrokh@math.kobe-u.ac.jp}}

\ArticleDates{Received October 10, 2017, in f\/inal form February 08, 2018; Published online February 16, 2018}

\Abstract{We consider the non-stationary Heun equation, also known as quantum Pain\-lev\'e~VI, which has appeared in dif\/ferent works on quantum integrable models and conformal f\/ield theory. We use a generalized kernel function identity to transform the problem to solve this equation into a dif\/ferential-dif\/ference equation which, as we show, can be solved by ef\/f\/icient recursive algorithms. We thus obtain series representations of solutions which provide elliptic generalizations of the Jacobi polynomials. These series reproduce, in a~limiting case, a~perturbative solution of the Heun equation due to Takemura, but our method is dif\/ferent in that we expand in non-conventional basis functions that allow us to obtain explicit formulas to all orders; in particular, for special parameter values, our series reduce to a single term.}


\Classification{33E20; 81Q05; 16R60}

\renewcommand{\thefootnote}{\arabic{footnote}}
\setcounter{footnote}{0}

\section{Introduction}
This paper is devoted to the study of a certain non-stationary Schr\"odinger equation with elliptic potential (see \eqref{Heun} below) which has appeared in dif\/ferent works on quantum integrable models and conformal f\/ield theory \cite{FLNO,Kolb,Ro}. This so-called {\em non-stationary Heun equation}~\cite{LT}, also known as {\em quantum Painlev\'e~VI}~\cite{Nag}, is a generalization of a second-order Fuchsian dif\/ferential equation with four regular singular points known as {\em Heun equation}~\cite{Ron,SL}. Another important special case is the {\em non-stationary Lam\'e equation} \cite{BM,EK}, which is also known as {\em KZB heat equation}\footnote{To be more precise: the special case of the KZB equations corresponding to $\mathfrak{g}=\mathfrak{sl}_2$ and $n=1$.}~\cite{FV}. Our main result is a method to construct a solution of the non-stationary Heun equation which is complete (in a sense explained below) and which is complementary to results obtained by other methods in that it can be used for generic parameter values. We also discuss various interesting special cases of our results. As will be explained, this is part of a research program using kernel functions to solve quantum models of Calogero--Moser--Sutherland (CMS) type; see \cite{HL,ELeCS0,ELJack,ELsigma,ELeCS2}.

\subsection{Background}
There is a long tradition in mathematical physics that regards integrable models and the mathematical theory of special functions as two sides of the same coin. While most textbook examples of such models can be solved using results about special functions known since a long time (they can be found in \cite{WW}, for example), modern developments in conformal f\/ield theory and quantum statistical physics have led to special functions which belong to function classes which are the subject of ongoing research. As an outstanding example we mention the representation theoretic approach to the solution of certain hyperbolic dif\/ferential equations due to Etingof, Frenkel and Kirillov \cite{EFK,EK}, which is a beautiful generalization of classical works on the Gauss hypergeometric equation pioneered by Gelfand and motivated by conformal f\/ield theory. A~complementary approach to this inspired by string theory was developed in works by Felder, Varchenko, and others \cite{FG,FV95,FV}. The non-stationary Lam\'e equation is the simplest non-trivial example to which these results apply. More recently this equation appeared in the construction of exact 4-point correlation functions of the quantum Liouville model~\cite{FLNO}, and in the exact solution of the eight-vertex model \cite{BM}. It was conjectured in~\cite{FLNO} that results about the non-stationary Lam\'e equation should have natural generalizations to the non-stationary Heun equation. This was recently conf\/irmed in important examples by Rosengren~\cite{Ro1,Ro}, who proved and generalized conjectures in~\cite{BM}, and Kolb \cite{Kolb}, who generalized the representation theoretic approach to the non-stationary Lam\'e equation in \cite{EFK,EK} to the Heun case. We also mention another important approach due to Nekrasov and Shatiashvili allowing to construct functions in the Heun class and which is based on supersymmetric gauge theories \cite{NS}; see also \cite{KS} for recent related work. Another related subject is the AGT conjecture which has led to explicit combinatorial expressions for conformal blocks related to the non-stationary Heun equation; see, e.g.,~\cite{Nawata} and references therein.

The Heun equation has received considerable interest as a natural equation def\/ining a function class generalizing the Gauss hypergeometric functions; see, e.g., \cite{Maier,Ron,SL,Smirnov,TakemuraHeun4} and references therein. The examples mentioned in the previous paragraph motivate us to extend this scope: our work is intended as a contribution towards a general theory of special functions def\/ined by the {\em non-stationary} Heun equation. The approach we use is based on a kernel function method developed in a series of papers in order to construct exact eigenfunctions of quantum models of CMS type; see \cite{HL,ELeCS0,ELJack,ELsigma,ELeCS2}. The relation of this to other approaches to the Heun equation based on kernel functions \cite{Erd,LaSl,LW,Novikov,RHeun} is discussed in Section~\ref{secFinal}.

\looseness=-1 To explain our method we recall that a kernel function for a pair of Schr\"odinger opera\-tors~$H(x)$ and $\tilde H(y)$ is an explicitly known function $K(x,y) $ satisfying the functional identity $(H(x)-\tilde H(y) -c)K(x,y)=0$  for some constant $c$; in the examples of interest to us the Schr\"odinger operators are Hamiltonians def\/ining a CMS-type model which can have equal ($H=\tilde H$) or dif\/ferent ($\tilde H\neq H$) parameters, and we write $H(x)$ to indicate that this dif\/ferential operator acts on functions depending on the variable $x$. A basic observation underlying the kernel function method we use is that CMS-type Hamiltonians have eigenfunctions which are easy to construct but uninteresting for applications, and that kernel functions provide a tool to transform such uninteresting eigenfunctions to interesting ones \cite{ELsigma}. It was shown in~\cite{HL} that all classical CMS-type models possess kernel functions which allow the construction of particular series representations of standard eigenfunctions in this way (by {\em classical CMS-type models} we mean those whose eigenfunctions provide natural many-variable generalizations of classical orthogonal polynomials, including the corresponding deformed CMS models in the sense of Chalykh, Feigin, Veselov and Sergeev~\cite{CFV,Sergeev}). We also mention recent related work of Halln\"as and Ruijsenaars~\cite{HR1} constructing eigenfunctions of the CMS model with $1/\sinh^2$-interactions using such a kernel function and which goes beyond the paradigm of polynomial eigenfunctions.
It was found that the elliptic generalizations of the CMS-models, which can be regarded as many-variable generalizations of the Lam\'e and Heun equations, possess kernel functions~\cite{ELrem,ELeCS1,LT} generalizing those in~\cite{HL}, and in~\cite{ELeCS0} one such kernel function was used to construct series representations of eigenfunctions of the elliptic quantum Calogero--Sutherland (eCS) model~\cite{ELeCS2}.
We mention in passing a series solution of the eCS model by Komori and Takemura~\cite{KT} which, while using the same expansion parameter, is dif\/ferent in important details from the result in~\cite{ELeCS0,ELeCS2} (this dif\/ference is analogous to the dif\/ference between perturbative results of Takemura on the Heun equation in~\cite{THeun} and results in the present paper, as discussed in the last paragraph of Section~\ref{secSummary1} below).

It is known that the kernel functions allowing for elliptic generalizations are restricted by a so-called {\em balancing condition} $\kappa=0$ with $\kappa$ a constant depending on the model parameters \cite{KNS,ELrem}.
If this condition is not fulf\/illed one can often f\/ind a {\em generalized kernel function} $K(x,y,\tau)$ satisfying
\begin{gather}\label{Kgen}
\left(\frac{{\rm i}}{\pi} \kappa \frac{\partial}{\partial\tau}+H(x,\tau)-\tilde H(y,\tau) -c(\tau)\right)K(x,y,\tau)=0
\end{gather}
with $\tau$ the half-period ratio of the elliptic functions appearing in the CMS Hamiltonians~$H$ and~$\tilde H$ \cite{ELrem,ELeCS1,LT} (see, e.g., Lemma~\ref{Lemma:kernel} where the $\tau$-dependence of~$H$, $\tilde H$, $c$ and $K$ is suppressed). In the present paper we use a generalized kernel function found in~\cite{LT} to solve the non-stationary Heun equation, following the approach in~\cite{ELeCS2}. The basic observation underlying our approach is that one can use the generalized kernel function in \eqref{Kgen} to transform an eigenfunction of the dif\/ferential operator $\frac{{\rm i}}{\pi} \kappa \frac{\partial}{\partial\tau}+\tilde H(y,\tau)$ to an eigenfunction of $\frac{{\rm i}}{\pi} \kappa \frac{\partial}{\partial\tau}+H(x,\tau)$ \cite{LT}.

\subsection{Summary of results}\label{secSummary1}
The non-stationary Heun equation can be written as\footnote{Here and in the following we often suppress the dependence of functions on the variable $\tau$.}
\begin{gather}\label{Heun}
\left( \frac{{\rm i}}{\pi}\kappa \frac{\partial}{\partial\tau} -\frac{\partial^2}{\partial x^2} +\sum_{\nu=0}^3 g_\nu(g_\nu-1) \wp(x+\omega_\nu)\right)\psi(x)=E\psi(x)
\end{gather}
with $\wp(x)$ the Weierstrass elliptic function with periods $(2\pi,2\pi\tau$) and
\begin{gather}\label{omeganu}
\omega_0= 0,\qquad \omega_1=\pi,\qquad \omega_2= -\pi-\pi\tau,\qquad \omega_3= \pi\tau
\end{gather}
(for the convenience of the reader we collect the def\/initions of $\wp$ and other well-known special functions we need in Appendix~\ref{appSpecialFunctions}). To simplify notation we set $\omega_1=\pi$ here and in most parts of this paper; see Appendix~\ref{appScaling} for how to transform our results to other values of $\omega_1$. The parameters $g_0$, $g_1$, $g_2$, $g_3$ and $\kappa$ can be arbitrary complex numbers for our general results.\footnote{We have to make some restrictions on parameters due to a technical problem referred to as {\em resonances} but, as discussed in Section~\ref{secResonances}, many of these restrictions are irrelevant in practice.} Our aim is to construct functions $\psi(x)\equiv \psi(x,\tau;\{g_\nu\}_{\nu=0}^3,\kappa)$ of two complex variables $x$ and $\tau$, $\Im(\tau)>0$, that satisfy this dif\/ferential equation for some $E\equiv E(\tau;\{g_\nu\}_{\nu=0}^3,\kappa)$; a more precise characterization of our solutions is given in \eqref{solution}--\eqref{P} below. It is important to note that, for $\kappa\neq 0$, $E$ can be transformed to 0, or any other convenient value, by changing the normalization of $\psi(x)$ (this follows from the obvious invariance of \eqref{Heun} under the transformation
 \begin{gather} \label{symmetry}
\psi(x)\to C\psi(x),\qquad E \to E + \frac{{\rm i}}{\pi}\kappa \frac1{C}\frac{\partial C}{\partial\tau}
\end{gather}
for arbitrary analytic functions $C$ of $\tau$). However, we sometimes f\/ind it convenient to impose a normalization condition on $\psi(x)$ so that $E$ remains signif\/icant even for $\kappa\neq 0$. Important special cases of \eqref{Heun} include the Heun equation ($\kappa=0$), the non-stationary Lam\'e equation ($g_\nu=g$ independent of $\nu$, or $g_\nu(g_\nu-1)= 0$ for three of the $\nu$'s), and the Lam\'e equation (both specializations). Many of our results are new even for these special cases (to our knowledge).

To explain the nature of our solutions we recall that, in the trigonometric limit $\Im(\tau)\to +\infty$, \eqref{Heun} reduces to the stationary Schr\"odinger equation with the P\"oschl--Teller potential $\propto g_0(g_0-1)/\sin^2(x/2) + g_1(g_1-1)/\cos^2(x/2)$ which has explicitly known solutions equal to Jacobi polynomials up to a common factor $\sin(x/2)^{g_0}\cos(x/2)^{g_1}$ (see \eqref{PT}--\eqref{PTsolution} for details).
The solutions of \eqref{Heun} that we construct are a generalization of this: they are of the form
\begin{gather}
\psi_n(x)= (2q^{1/4})^{-g_0-g_1}\left( \prod_{\nu=0}^3 \theta_{\nu+1}\big(\tfrac{1}{2} x\big)^{g_\nu} \right)\mathcal{P}_n\left(\cos(x) \right),\nonumber\\
E_n = \kappa^2\left(\frac{1}{12} -\frac{\eta_1}{\pi}\right) -\sum_{\nu=0}^3 g_\nu(g_\nu-1)\frac{\eta_1}{\pi} + \mathcal{E}_n\label{solution}
\end{gather}
with $\theta_{\nu+1}(z)$ the Jacobi theta functions, $q=\exp({\rm i}\pi\tau)$ the nom\'e, $\eta_1/\pi$ in \eqref{eta1pi} and
\begin{gather}\label{P}
\mathcal{P}_n(z) = \sum_{\ell=0}^\infty \mathcal{P}^{(\ell)}_n(z) q^\ell ,\qquad \mathcal{E}_n = \sum_{\ell=0}^\infty \mathcal{E}^{(\ell)}_nq^\ell,
\end{gather}
with
\begin{gather}\label{eJacobi}
\mathcal{P}^{(0)}_n(z) = P_n^{\big(g_0-\frac{1}{2},g_1-\frac{1}{2}\big)}(z),\qquad \mathcal{E}^{(0)}_n = \left(n+\frac{g_0+g_1}{2}\right)^2,\qquad n\in\mathbb{N}_0,
\end{gather}
and $P^{\big(g_0-\frac12,g_1-\frac12\big)}_n(z)$ the Jacobi polynomials (see~\eqref{Series}). Our main result provides ef\/f\/icient recursive procedures to compute the functions $\mathcal{P}^{(\ell)}_n(z)$ which, as we show, are polynomials of degree $n+\ell$ in $z$; see Propositions~\ref{prop1} and~\ref{prop2} for two complementary variants of this result. Thus one can regard~$\mathcal{P}_n(z)$ in~\eqref{P} as an elliptic generalization of Jacobi polynomials. It is interesting to note that these elliptic generalizations exist even for negative integers~$n$ if~$q$ is non-zero, but they vanish like~$O(q^{-n})$ for $n<0$ as $q\to 0$. We conjecture that the series in~\eqref{P} is absolutely convergent and converges to a $L^2$-function on $[0,\pi]$ for $|q|\leq q_0$ and some $q_0>0$ depending on parameters (this is known to be true in the Heun case $\kappa=0$ from work by Takemura \cite{THeun}; see also \cite{ELeCS2} for a convergence proof for the Lam\'e case which, as we believe, can be generalized). However, this question is left for future work: for simplicity we treat series like in~\eqref{P} as formal power series.

Our solution \eqref{solution}--\eqref{P} of \eqref{Heun} is complete in the sense that the $\psi_n(x)$ provide a complete orthonormal basis in the Hilbert space of $L^2$-functions on $[0,\pi]$ for $g_0+g_1>0$ in the trigonometric case $q=0$ (we believe that this is true even for $q > 0$).
Moreover, we give an ef\/f\/icient recursive procedure to compute the coef\/f\/icients $ \mathcal{P}^{(\ell)}_n(z)$ and $ \mathcal{E}^{(\ell)}_n$ of the power series in~\eqref{P}. We note that, in the Heun case $\kappa=0$, the $E_n$ correspond to the eigenvalues of the $BC_1$ elliptic CMS Hamiltonian discovered by Inozemtsev~\cite{I}. We thus refer to the $E_n$ as {\em generalized eigenvalues} in the following.

It is important to note that, by exploiting the invariance of \eqref{Heun} under the transformation in \eqref{symmetry}, we obtain two complementary variants of results: in the f\/irst variant we impose a~normalization conditions on $\psi_n(x)$ such that the generalized eigenvalues $E_n$ are signif\/icant (see Proposition~\ref{prop1} and Theorem~\ref{Thm1}), and in the second we f\/ix $E_n$ by a convenient condition (see Proposition~\ref{prop2} and Theorem~\ref{Thm2}). These two variants of results are complementary to each other in that the second is somewhat simpler but restricted to $\kappa\neq 0$, whereas the f\/irst applies to the case $\kappa=0$ as well. However, as will be discussed after Proposition~\ref{prop2}, the second variant of results implies an interesting representation of the eigenvalues $E_n$ in the limit $\kappa\to 0$.

One important feature of our method is that it provides particular $q$-dependent basis functions $\{f_m(z)\}_{m\in\mathbb{Z}}$ to expand the functions $\mathcal{P}_n(z)$ in; see~\eqref{fn} and~\eqref{Pnseries}. This is useful since these functions $f_m(z)$ take into account much of the complexity of the problem; for example, in special cases the expansion coef\/f\/icients are trivial, and in these cases our method gives explicit integral representations of the solutions $\mathcal{P}_n(z)$ (see Section~\ref{secExplicit}). For general parameter values, we obtain a system of equations for these expansion coef\/f\/icients which, in the trigonometric case $q=0$, can be solved by diagonalizing a triangular matrix and which, for non-zero $q$, can be solved by ef\/f\/icient perturbative algorithms; see Propositions~\ref{prop1} and~\ref{prop2}. A perturbative solution of these algorithms to all orders in $q$ is obtained in Section~\ref{secAllOrders}; see Theorems~\ref{Thm1} and~\ref{Thm2}. We note that results for the elliptic Calogero--Sutherland model corresponding to the ones in Sections~\ref{secTrigonometric}, \ref{secPerturbative} and \ref{secAllOrders} were obtained by one of us in \cite{ELeCS0,ELJack}, and \cite{ELeCS2}, respectively. We also mention that, in the special case $\kappa=0$, we obtain results for the Heun equation in Section~\ref{secAllOrders0} which dif\/fer from the ones by Takemura who used Jacobi polynomials as basis to expand the functions $\mathcal{P}_n(z)$~\cite{THeun}. As already mentioned, this dif\/ference is analogous to the dif\/ference between the perturbative results for the eCS model in~\cite{ELeCS0} and the one by Komori and Takemura in \cite{KT}: in the latter work the eigenfunctions are expanded in Jack polynomials, whereas in the former an unconventional basis is used which allows for an explicit solution to all orders~\cite{ELeCS2}.

\subsection{Plan}
Section~\ref{secPreliminaries} contains preliminary material: a summary of notation (Section~\ref{secNotation}), a~review of a~well-known solution of \eqref{Heun} for $\Im(\tau)\to\infty$ in terms of Jacobi polynomials (Section~\ref{secPT}), the def\/inition and properties of our basis functions $f_m(z)$ (Section~\ref{secfm}), and a discussion of a technicality referred to as {\em resonances} (Section~\ref{secResonances}). In Section~\ref{secResults} we present our key result, which is a transformation of the problem to solve \eqref{Heun} into a~dif\/ferential-dif\/ference equation (Proposition~\ref{prop0}), together with a discussion of special cases (Section~\ref{secExplicit}); the proof of this key result is given in Section~\ref{Proofprop0}.
In Section~\ref{secPerturbative} we present two complementary recursive algorithms to solve this dif\/ferential-dif\/ference equation, and Section~\ref{secAllOrders} contains the corresponding explicit solutions to all orders.
We conclude with f\/inal remarks in Section~\ref{secFinal}.
Five appendices contain def\/initions and properties of special functions we use (Appendix~\ref{appSpecialFunctions}), details on how to translate our results for $\omega_1=\pi$ to other values of $\omega_1$ (Appendix~\ref{appScaling}), explicit results from one of our recursive algorithms at low order (Appendix~\ref{appExplicitResults}), derivations of results needed in proofs (Appendix~\ref{appComputations}), and a short discussion of the combinatorial structure of our solution (Appendix~\ref{appCombinatorics}).

\section{Preliminaries}\label{secPreliminaries}
We collect def\/initions and preliminary results that we use.

\subsection{Notation}\label{secNotation}
We use the special functions ($\xi$ and $q$ are complex variables; $|q|<1$)
\begin{gather}
\Theta_1(\xi) \equiv (1-\xi)\prod_{n=1}^\infty\big(1-q^{2n}\xi\big)\big(1-q^{2n}\xi^{-1}\big),\qquad \Theta_2(\xi)\equiv \Theta_1(-\xi),\nonumber\\
\Theta_3(\xi) \equiv \prod_{n=1}^\infty\big(1+q^{2n-1}\xi\big)\big(1+q^{2n-1}\xi^{-1}\big),\qquad \Theta_4(\xi)\equiv \Theta_3(-\xi)\label{Thetanu}
\end{gather}
and
\begin{gather}\label{Theta}
\Theta(z,\xi)\equiv \big(1-2z\xi+\xi^2\big)\prod_{n=1}^\infty\big(1-2q^{2n}\xi z+q^{4n}\xi^2\big)\big(1-2q^{2n}\xi^{-1}z+q^{4n}\xi^{-2}\big),
\end{gather}
which all are closely related to the Jacobi theta functions (see~\eqref{tettet} and~\eqref{tettet1}).

We denote as $\mathbb{N}_0$ and $\mathbb{Z}'$ the sets of non-negative and non-zero integers, respectively. The symbol $\delta(m,n)$ for integers $m$, $n$ denotes the Kronecker delta.

\subsection{Trigonometric limit}\label{secPT}
The non-stationary Heun equation simplif\/ies in the trigonometric case $q=0$ to
\begin{gather}\label{PT}
\left( -\frac{\partial^2}{\partial x^2} + \frac{g_0(g_0-1)}{4\sin^2 \frac{1}{2} x} + \frac{g_1(g_1-1)}{4\cos^2 \frac{1}{2} x} \right)\psi^{(0)}(x)=E^{(0)}\psi^{(0)}(x) ,
\end{gather}
which is known to have solutions
\begin{gather}\label{PTsolution}
\psi_n^{(0)}(x)=\big(\sin \tfrac{1}{2} x\big)^{g_0}\big(\cos \tfrac{1}{2} x\big)^{g_1}P^{\big(g_0-\frac12,g_1-\frac12\big)}_n(\cos x), \qquad
E^{(0)}_n = \left(n+\frac{g_0+g_1}2\right)^2
\end{gather}
with the Jacobi polynomials $P_n^{(\alpha,\beta)}(z)$ in \eqref{Series} (see, e.g., \cite[Table~18.8.1, 2nd line]{Dig10}). We will use a well-known uniqueness result about this solution in the following form.

\begin{Lemma}\label{lemmaUniqueness}
Let $\psi_n^{(0)}(x)$ be a solution of \eqref{PT} of the form
\begin{gather*}
\psi_n^{(0)}(x) = \big(\sin \tfrac{1}{2} x\big)^{g_0}\big(\cos \tfrac{1}{2} x\big)^{g_1}P(\cos x)
\end{gather*}
for some constant $E^{(0)}$, with $P(z)$ a polynomial of degree $n$ such that
\begin{gather}\label{Pnormalization}
P(z) = \frac{(n+g_0+g_1)_n}{2^nn!}z^n+O\big(z^{n-1}\big)
\end{gather}
for some $n\in\mathbb{N}_0$. Then $P(z)=P^{\big(g_0-\frac12,g_1-\frac12\big)}_n(z)$ and $E^{(0)} =E^{(0)}_n$.
\end{Lemma}

The proof is standard and therefore omitted. (Note that the normalization in \eqref{Pnormalization} follows from~\eqref{JacobiNormalization}.)

We will use this result to f\/ix the normalization of our solutions so as to get Jacobi polynomials in the trigonometric case $q=0$.

\subsection{Basis functions}\label{secfm}
As mentioned in the introduction, one important feature of our method is that we use non-trivial basis functions $f_m(z)$. These functions are def\/ined by the following generating function,
\begin{gather}\label{fn}
\frac{\prod\limits_{\nu=0}^3\Theta_{\nu+1}(\xi)^{\tilde{g}_\nu}}{\Theta(z,\xi)^\lambda}\equiv \sum_{m\in \mathbb{Z}}f_m(z)\xi^{m},\qquad |q|<|\xi|<1,
\\ \label{lambda}
\tilde g_\nu\equiv \lambda-g_\nu,\qquad \lambda \equiv \tfrac12(g_0+g_1+g_2+g_3-\kappa)
\end{gather}
with the special functions $\Theta_{\nu+1}(\xi)$, $\Theta(z,\xi)$ def\/ined in \eqref{Thetanu}--\eqref{Theta}. It is easy to check that the series on the r.h.s.\ in~\eqref{fn} is absolutely convergent in the region indicated, and thus the functions~$f_m(z)$ are well-def\/ined.\footnote{One can show that, for f\/ixed complex parameters $\{g_\nu\}_{\nu=0}^3$ and~$\kappa$, and all $m\in\mathbb{Z}$, $f_m(z)$ is analytic for $|q|<1$ and $z$ in some $q$-dependent open domain which includes the interval $[-1,1]$.}  Since we restrict ourselves to results in the sense of formal power series in $q$, we only need the following characterization of these functions (the proof of this is technical and thus deferred to an appendix).

\begin{Lemma}\label{Lemma:fn2}
The functions $f_m(z)$ defined in \eqref{fn}--\eqref{lambda} have the following power series expansion
\begin{gather}\label{fnseries}
f_m(z) = \sum_{\ell=0}^\infty f_m^{(\ell)}(z)q^\ell
\end{gather}
with $f_m^{(\ell)}(z)=0$ for $m+\ell<0$ and $f_m^{(\ell)}(z)$ a polynomial of degree $m+\ell$ in $z$ for $m+\ell\geq 0$. In particular,
\begin{gather}\label{f0n}
f_m^{(0)}(z) = \binom{-\lambda}{m}(-2z)^m + O\big(z^{m-1}\big)
\end{gather}
with the binomial coefficient $\binom{-\lambda}{m}$ as usual $($see \eqref{binomial}$)$.
\end{Lemma}

\begin{proof} See Appendix~\ref{secProofLemfn2}. \end{proof}

We will also use the following integral representation of these functions:
\begin{gather}\label{fnint}
f_m(z) = \oint_{\mathcal{C}}\frac{d\xi}{2\pi{\rm i}\xi}\xi^{-m} \frac{\prod\limits_{\nu=0}^3\Theta_{\nu+1}(\xi)^{\tilde g_\nu}}{\Theta(z,\xi)^\lambda}
\end{gather}
with $\mathcal{C}$ the contour once around the circle with radius $|\xi|=R$, $|q|<R<1$, taken counterclockwise (this is equivalent to \eqref{fn} by Cauchy's theorem).

\begin{Remark}\label{remlambdacondition}
It is clear from \eqref{f0n} that the functions $f^{(0)}_m(z)$ are non-trivial polynomials in~$z$ of degree $m$ for all $m\in\mathbb{N}_0$ only if $-\lambda\notin\mathbb{N}_0$ (i.e., only in this case the binomial coef\/f\/icient on the r.h.s.\ in~\eqref{f0n} is always non-zero). If $-\lambda=k\in\mathbb{N}_0$, one can see from \eqref{fnint} that $f^{(0)}_m(z)$ is a~polynomials of degree $\leq k$ for all $m\in\mathbb{N}_0$. Thus, to get complete results, we sometimes impose the condition $-\lambda\notin\mathbb{N}_0$.
\end{Remark}

\subsection{Resonances}\label{secResonances}
We discuss a technical issue encountered in Sections~\ref{secPerturbative} and \ref{secAllOrders}: to prove Propositions~\ref{prop1} and~\ref{prop2} and Theorems~\ref{Thm1} and \ref{Thm2} we f\/ind it convenient to impose the following {\it no-resonance conditions:} {\em either $\Im(\kappa)\neq 0$ and $g_0+g_1\in\mathbb{R}$, or $\kappa=0$ and $g_0+g_1\notin\mathbb{Z}$}. At f\/irst sight this seems to exclude many cases of interests in applications but, at closer inspection, one f\/inds that this is not the case: as explained in this section, our results can be used even in cases where these conditions fail.

We f\/irst explain the reason for these conditions: our solutions are obtained by an unconventional variant of perturbation theory leading to series solutions which are linear combinations of products of the following generalized energy dif\/ference fractions,
\begin{gather}\label{resonance1}
\frac1{E^{(0)}_{m}-E^{(0)}_n -\kappa\ell} = \frac1{(m-n)(m+n+g_0+g_1)-\kappa\ell}
\end{gather}
with $E^{(0)}_n$ in \eqref{PTsolution} the energy eigenvalues of the unperturbed problem $q=0$ and $\ell=0,1,2,\ldots$; we refer to a case $(m,\ell)$ where, for f\/ixed $n$, the denominator in \eqref{resonance1} is zero as {\em resonance}. The reason for the conditions on parameters above is that they are a simple means to rule out resonances, and this guarantees that our series are well-def\/ined.
However, while these conditions are suf\/f\/icient, they are not necessary: one peculiar feature of our perturbation theory is that the series we obtain have singularities coming from energy dif\/ference denominators as in \eqref{resonance1}, but many of these singularities are removable. Thus our results can be extended to cases where our no-resonance conditions fail. We now give a precise formulation for one important such case.

\begin{Lemma}\label{lemResonance}
In the Heun case $\kappa=0$ and for fixed $n\in\mathbb{N}_0$, the results for the expansion coefficients $\mathcal{P}_n^{(\ell)}(z)$ and $\mathcal{E}_n^{(\ell)}$ in Proposition~{\rm \ref{prop1}} and Theorem~{\rm \ref{Thm1}} are valid even in the limit $g_0+g_1\to k\in\mathbb{Z}$ with $k>-(2n+1)$.
\end{Lemma}
(The proof is given at the end of this section.)

Thus, in the Heun case $\kappa=0$ and for all $n\in\mathbb{N}_0$, our results can be extended to the case $g_0+g_1\in\mathbb{N}_0$ of interest in applications. We believe that, in a similar manner, our results can be extended to interesting cases with non-zero {\em real} $\kappa$. As will become clear in our proof of Lemma~\ref{lemResonance} below, the challenge to make precise and prove such a result is that the generalization of standard perturbation theory~\cite{THeun} to $\kappa\neq 0$ is not known (to our knowledge).

\begin{proof}[Proof of Lemma~\ref{lemResonance}]
For $\kappa=0$ and f\/ixed $n\in\mathbb{N}_0$, the energy denominator in \eqref{resonance1} appearing in standard perturbation theory \cite{THeun} are only for $m\in\mathbb{N}_0$ dif\/ferent from $n$, and thus all pertinent fractions in \eqref{resonance1} are f\/inite if $g_0+g_1>-(2n+1)$.
In our perturbative expansions we encounter fractions in \eqref{resonance1} for arbitrary integers $m\neq n$.
However, it is clear that our results for the coef\/f\/icients $\mathcal{P}_n^{(\ell)}(z)$ and $\mathcal{E}_n^{(\ell)}$ of the perturbative solution def\/ined in \eqref{P} must be identical with the corresponding results obtained in standard perturbation theory. Thus the singularities coming from resonance fractions in our perturbation expansion must cancel: our perturbative results remain f\/inite even in the limit when $g_0+g_1$ becomes integer.
\end{proof}

\begin{Remark}The resonance problem that we encounter in this paper is very similar to the one which appeared in the treatment of the eCS model in \cite{ELeCS1,ELeCS2}.
This is no coincidence: results in the special case $N=2$ in {\em op.\ cit.}\ correspond to ours in the Lam\'e case $g_0=g_1=g_2=g_3$, $\kappa=0$.
The interested reader is referred to \cite{ELeCS2} for a more extensive discussion of resonances.
\end{Remark}

\section{Key result and special cases}\label{secResults}
In Section~\ref{secKey} we present our key result, which is a transformation of the problem to solve the non-stationary Heun equation in \eqref{Heun} with the ansatz in \eqref{solution} to a problem to solve a~dif\/ferential-dif\/ference equation; as we show in subsequent sections, the latter problem allows for ef\/f\/icient solutions. In Section~\ref{secExplicit} we point out special non-trivial cases where our key result leads to explicit integral representations of solutions of \eqref{Heun}.

\subsection{Dif\/ferential-dif\/ference equation}\label{secKey}
We construct solutions $\psi_n(x)$, $E_n$ of the non-stationary Heun equation in \eqref{Heun} of the form \eqref{solution}--\eqref{eJacobi} in the sense of formal power series in $q$.
One important feature of our method is that we expand
\begin{gather}\label{Pnseries}
\mathcal{P}_n(z) = {\mathcal N}_n \sum_{m \in\mathbb{Z}} \alpha_n(m)f_{m}(z)
\end{gather}
with non-trivial basis functions $f_m(z)$ given in \eqref{fn}--\eqref{lambda} and characterized in Lemma~\ref{Lemma:fn2}.
As will be shown, the following constant ensures the normalization in \eqref{eJacobi},
\begin{gather}\label{cNn}
{\mathcal N}_n = \frac{(n+g_0+g_1)_n}{4^n(\lambda)_n}
\end{gather}
with the raising Pochhammer symbol $(x)_n$ in \eqref{Pochhammer} and $\lambda$ in \eqref{lambda}; note that ${\mathcal N}_n$ is f\/inite and non-zero for all integers $n$ if $-\lambda\notin\mathbb{N}_0$ for $n>0$ and $-(g_0+g_1)\notin\mathbb{N}_0$ for $n<0$ (see the discussion after \eqref{Pochhammer}).

Our key result is equations determining $\alpha_n(m)$ and $\mathcal{E}_n$ and which, as we will show, can be solved ef\/f\/iciently.
To state this result we introduce the convenient shorthand notation
\begin{gather}\label{gammanu}
\gamma_k^\mu\equiv \begin{cases} \tilde g_0(\tilde g_0-1) + (-1)^\mu \tilde g_1(\tilde g_1-1) & \text{if }  \frac{1}{2} k\in\mathbb{N}_0, \\
(-1)^\mu \tilde g_2(\tilde g_2-1) + \tilde g_3(\tilde g_3-1) & \text{if } \frac{1}{2} (k-1)\in\mathbb{N}_0 \end{cases}
\end{gather}
for $(\mu,k)\in\mathbb{Z}\times \mathbb{N}_0$ (recall that $\tilde{g}_0=\frac12(-g_0+g_1+g_2+g_3-\kappa)$, $\tilde{g}_1=\frac12(g_0-g_1+g_2+g_3-\kappa)$ etc.; note that $\gamma_k^\mu=\gamma_{k+2r}^{\mu+2s}$ for all integers~$r$,~$s$).

\begin{Proposition}\label{prop0}
Let $n\in\mathbb{Z}$, $-\lambda\notin\mathbb{N}_0$ for $n>0$ and $-(g_0+g_1)\notin\mathbb{N}_0$ for $n<0$, and assume that $\mathcal{E}_n$ and $\alpha_n(m)$ for $m\in\mathbb{Z}$ satisfy the following system of equations,
\begin{gather}
\left[ -\kappa q\frac{\partial}{\partial q} +E^{(0)}_{m}-\mathcal{E}_n\right]\alpha_n(m)\nonumber\\
\qquad {} =\sum_{\mu=1}^\infty\mu \gamma_0^\mu \alpha_n(m+\mu)
+ \sum_{\mu=1}^\infty \frac{\mu q^\mu}{1-q^{2\mu}}\big( \gamma_0^\mu q^{\mu} + \gamma_1^\mu \big) [\alpha_n(m+\mu)+\alpha_n(m-\mu)]\label{aneqs}
\end{gather}
and the condition
\begin{gather}\label{ic}
\alpha_n(m)|_{q=0} = \begin{cases} 0, & m>n, \\ 1, & m=n. \end{cases}
\end{gather}
Then $\psi_n(x)$, $E_n$ in \eqref{solution}--\eqref{P} and \eqref{Pnseries}--\eqref{cNn} satisfy the non-stationary Heun equation in \eqref{Heun}, and the conditions in \eqref{eJacobi} hold true provided $-(g_0+g_1)\notin\mathbb{N}$.
\end{Proposition}
(The proof is given in Section~\ref{Proofprop0}.)

It is important to note that the conditions above do not determine $\mathcal{P}_n(z)$ and $\mathcal{E}_n$ uniquely: \eqref{aneqs}--\eqref{ic} are invariant under
\begin{gather}\label{invariance}
\alpha_n(m)\to C   \alpha_n(m),\qquad \mathcal{E}_n\to \mathcal{E}_n - \kappa \frac1{C} q\frac{\partial}{\partial q}C
\end{gather}
for any change of normalization $C=1+O(q)$ analytic in $q$ (this is a consequence of the invariance of~\eqref{Heun} under~\eqref{symmetry}).
This ambiguity can be f\/ixed for generic parameter values by replacing~\eqref{ic} by a stronger condition; see Propositions~\ref{prop1} and~\ref{prop2} for two dif\/ferent ways to do this.

This result also provides simple explicit solutions of \eqref{Heun} for special particular parameter values, as elaborated in Section~\ref{secExplicit}.

\begin{Remark}\label{remP}
It interesting to note that the solutions $\alpha_n(m)$ and $\mathcal{E}_n$ of the equations in Proposition~\ref{prop0} are essentially independent of $n$ in the following sense: they are of the form
\begin{gather*}
\alpha_n(m)=a(m-n) ,\qquad \mathcal{E}_n = (P/2)^2 +\tilde{\mathcal{E}}
\end{gather*}
with functions $a(k)$ and $\tilde{\mathcal{E}}$ depending on $n$ only in the combination
\begin{gather}\label{Pdef}
P \equiv 2n+g_0+g_1
\end{gather}
(this is easy to check).
This and the notation introduced here are useful in computations and in the presentation of results; we use this in \eqref{cEn1}--\eqref{cEn2} and in Appendix~\ref{appExplicitResults}.
\end{Remark}

\subsection{Explicit solutions by integrals}\label{secExplicit}
The solutions $\alpha_n(m)$, $\mathcal{E}_n$ of \eqref{aneqs}--\eqref{ic} are complicated in general.
However, there exist non-trivial cases where Proposition~\ref{prop0} gives simple explicit solutions of the non-stationary Heun equation:

\begin{Corollary}\label{corSimple}
Let $n\in\mathbb{Z}$, $\lambda$ a complex parameter such that $-\lambda\notin\mathbb{N}_0$ for $n>0$ and $-(g_0+g_1)\notin\mathbb{N}_0$ for $n<0$,
\begin{gather}\label{parametersgnu}
\tilde g_\nu\in\{0,1\},\qquad g_\nu=\lambda-\tilde g_\nu,\qquad \nu=0,1,2,3,\qquad \kappa=2\lambda-\sum_{\nu=0}^3\tilde g_\nu,
\end{gather}
${\mathcal N}_n$ in \eqref{cNn}, and $\mathcal{C}$ the integration contour defined after \eqref{fnint}. Then the non-stationary Heun equation in \eqref{Heun} has solutions $\psi_n(x)$, $E_n$ as in \eqref{solution} with
\begin{gather}\label{ExplicitSolution}
\mathcal{P}_n(z) = {\mathcal N}_n\oint_{\mathcal{C}}\frac{d\xi}{2\pi{\rm i}\xi}\xi^{-n} \frac{\prod\limits_{\nu=0}^3\Theta_{\nu+1}(\xi)^{\tilde g_\nu}}{\Theta(z,\xi)^\lambda}, \qquad
\mathcal{E}_n = \left(n+\frac{g_0+g_1}2\right)^2
\end{gather}
and such that the conditions in \eqref{eJacobi} hold true provided $-(g_0+g_1)\notin\mathbb{N}$.
\end{Corollary}
\begin{proof}
The assumptions imply $\gamma_k^\mu=0$ for all $k$, $\mu$, and the equations in \eqref{aneqs}--\eqref{ic} in this case have the solution $\alpha_n(m)=\delta(m,n)$, $\mathcal{E}_n=E^{(0)}_n$. Proposition~\ref{prop0} implies the result.
\end{proof}

Note that \eqref{parametersgnu} gives several dif\/ferent one-parameter families $(\{g_\nu\}_{\nu=0}^3,\kappa)$, depending on $\lambda$, where this result provides simple integral representations of elliptic Jacobi polynomials $\mathcal{P}_n(z)$. Moreover, these formulas are non-trivial even in the trigonometric case $q=0$: the results above imply the following integral representations of Jacobi polynomials,
\begin{gather}
P_n^{\big(g-\frac{1}{2},g-\frac{1}{2}\big)}(z) =  \frac{(n+2g)_n}{4^n(g)_n} \oint_{\mathcal{C}}\frac{d\xi}{2\pi{\rm i}\xi}\xi^{-n} \frac1{(1-2z\xi+\xi^2)^g}, \nonumber\\
P_n^{\big(g-\frac{1}{2},g+\frac{1}{2}\big)}(z) =  \frac{(n+2g+1)_n}{4^n(g+1)_n} \oint_{\mathcal{C}}\frac{d\xi}{2\pi{\rm i}\xi}\xi^{-n} \frac{(1-\xi)}{(1-2z\xi+\xi^2)^{g+1}},\nonumber\\
P_n^{\big(g+\frac{1}{2},g-\frac{1}{2}\big)}(z) =  \frac{(n+2g+1)_n}{4^n(g+1)_n} \oint_{\mathcal{C}}\frac{d\xi}{2\pi{\rm i}\xi}\xi^{-n} \frac{(1+\xi)}{(1-2z\xi+\xi^2)^{g+1}},\nonumber\\
P_n^{\big(g-\frac{1}{2},g-\frac{1}{2}\big)}(z) =  \frac{(n+2g)_n}{4^n(g+1)_n} \oint_{\mathcal{C}}\frac{d\xi}{2\pi{\rm i}\xi}\xi^{-n} \frac{(1-\xi^2)}{(1-2z\xi+\xi^2)^{g+1}}\label{Simple1}
\end{gather}
for $n\in\mathbb{N}_0$ (the latter identities are obtained from Corollary~\ref{corSimple} for the cases where $(\tilde{g}_0,\tilde{g}_1,\lambda)$ is $(0,0,g)$, $(1,0,g+1)$, $(0,1,g+1)$, and $(1,1,g+1)$, respectively). The f\/irst identity in \eqref{Simple1} is equivalent to a well-known generating function for the Gegenbauer polynomials (see \eqref{Gegenbauer}), and also the others can be found in \cite{Dig10}.

\begin{Remark}\label{RelationToFLNO}
Fateev et al.\ gave integral representations of solutions of the non-stationary Lam\'e equation \cite{FLNO} which, in a special case, are similar to the one above for $g_0=g_1=g_2=g_3$. More specif\/ically, the solution given in equation~(3.11) of~\cite{FLNO} can be proved by a simple variant of the argument that we used in order to obtain our solution in~\eqref{ExplicitSolution} (note that our parameter $\kappa$ corresponds to $-2/b^2$ in~\cite{FLNO}). We also mention similar integral representations of solutions of the non-stationary Lam\'e equation appearing in works by Etingof and Kirillov (see \cite[Theorem~5.1]{EK}) and Felder and Varchenko (see, e.g., \cite[Example~1.2]{FV}).
\end{Remark}

We emphasize that the result in Corollary~\ref{corSimple} is more general in that it includes some non-stationary Heun cases that cannot be reduced to a non-stationary Lam\'e case.

\begin{Remark}Corollary~\ref{corSimple} can be obtained as special case $(N,\tilde{N})=(1,0)$ from Proposition~4.1 in~\cite{LT}. However, this is not easy to see, and it is therefore worthwhile to emphasize this result here.
\end{Remark}

\section{Proof of key result}\label{Proofprop0}
We turn to the proof of Proposition~\ref{prop0}. In Section~\ref{secKernelMethod} we derive the key identity using the kernel function method. In Section~\ref{secTrigonometric} we consider the trigonometric case $q=0$ to prove that the conditions in \eqref{ic} and the normalization condition in \eqref{cNn} yield a solution satisfying \eqref{eJacobi}.

\subsection{Kernel function method}\label{secKernelMethod}
We introduce the notation
\begin{gather}\label{H}
H\big(x;\{ g_\nu\}_{\nu=0}^3\big) \equiv -\frac{\partial^2}{\partial x^2} +\sum_{\nu=0}^3 g_\nu(g_\nu-1)\wp(x+\omega_\nu)
\end{gather}
with $\omega_\nu$ in \eqref{omeganu}. This allows us to write the non-stationary Heun equation in \eqref{Heun} as
\begin{gather}\label{Heun1}
\left( \frac{{\rm i}}{\pi}\kappa \frac{\partial}{\partial\tau} + H\big(x;\{g_\nu\}_{\nu=0}^3\big) - E\right)\psi(x)=0.
\end{gather}
We also recall the def\/initions of $\tilde g_\nu$ and $\lambda$ in \eqref{lambda}. Note that $H$ in \eqref{H} is the Hamiltonian def\/ining the $BC_1$ Inozemtsev model \cite{THeun}.

We obtain our result from the following generalized kernel function identity:

\begin{Lemma}\label{Lemma:kernel}
The function
\begin{gather}\label{K}
\mathcal{K}(x,y)\equiv
\frac{\prod\limits_{\nu=0}^3\theta_{\nu+1}\big(\frac{1}{2} x\big)^{g_\nu}\theta_{\nu+1}\big(\frac{1}{2} y\big)^{\tilde g_\nu}}{\theta_1\big(\frac{1}{2} (x+y)\big)^{\lambda}\theta_1\big(\frac{1}{2} (x-y)\big)^{\lambda}}
\end{gather}
obeys the identity
\begin{gather}\label{kernel}
\left(\frac{{\rm i}}{\pi} \kappa\frac{\partial}{\partial\tau} + H\big(x;\{g_\nu\}_{\nu=0}^3\big) - H\big(y;\{\tilde g_\nu\}_{\nu=0}^3\big)-C_{1,1} \right) \mathcal{K}(x,y)=0
\end{gather}
with
\begin{gather}\label{C11}
C_{1,1} = 2\kappa (1-\lambda)\frac{\eta_1}{\pi}.
\end{gather}
\end{Lemma}

\begin{proof}This is the special case $N=M=1$ of Corollary~3.2 in \cite{LT} (note that the symbols $\beta$ and $A_{1,1}$ in \cite{LT} correspond to $-2\pi{\rm i} \tau$ and $2\kappa$ here, respectively).
\end{proof}

We use this to compute the action of $\frac{{\rm i}}{\pi} \kappa \frac{\partial}{\partial\tau} + H(x;\{g_\nu\}_{\nu=0}^3)$ on the functions
\begin{gather}\label{Fn}
F_m(x)\equiv \big(2q^{1/4}\big)^{-(g_0+g_1)}\left( \prod_{\nu=0}^3 \theta_{\nu+1}\big(\tfrac{1}{2} x\big)^{g_\nu} \right) f_m(\cos x),\qquad m\in\mathbb{Z}
\end{gather}
with $f_m(z)$ def\/ined in \eqref{fn}. For this we note that
\begin{gather*}
\mathcal{K}(x,y) = 2^{g_0+g_1}{\rm e}^{{\rm i} \pi\tilde g_0/2}G^{-\kappa }\big(2q^{1/4}\big)^{-(g_0+g_1)}\left( \prod_{\nu=0}^3 \theta_{\nu+1}\big(\tfrac{1}{2} x\big)^{g_\nu} \right)\frac{\prod\limits_{\nu=0}^3\Theta_{\nu+1}({\rm e}^{{\rm i} y})^{\tilde g_\nu}}{\Theta(\cos x,{\rm e}^{{\rm i} y})^\lambda} {\rm e}^{{\rm i} y(g_0+g_1)/2}
\end{gather*}
with $G$ def\/ined in \eqref{G} (we used \eqref{tettet} and \eqref{tettet1}). This and \eqref{fn} show that $\mathcal{K}(x,y)$ is a~generating function for the functions in \eqref{Fn}:
\begin{gather}\label{genfun}
 \mathcal{K}(x,y) = 2^{g_0+g_1}{\rm e}^{{\rm i} \pi\tilde g_0/2}G^{-\kappa} \sum_{m\in\mathbb{Z}} F_m(x) {\rm e}^{{\rm i} \big(m+\frac{1}{2}(g_0+g_1)\big)y}, \qquad 0<\Im(y)<\pi\Im(\tau).
\end{gather}
To evaluate $H(y;\{\tilde g_\nu\}_{\nu=0}^3) \mathcal{K}(x,y)$ we use the following expansions
\begin{gather}\label{wpseries}
\wp(y+\omega_\nu) = - \frac{\eta_1}{\pi} - \sum_{\mu\in{\mathbb{Z}'}} (S_\nu)_\mu{\rm e}^{{\rm i}\mu y},\qquad 0<\Im(y)<\pi\Im(\tau)
\end{gather}
with
\begin{gather}
(S_0)_\mu = \mu\frac{1}{1-q^{2\mu}}= |\mu|\frac{q^{|\mu|-\mu}}{1-q^{2|\mu|}},\nonumber\\
(S_1)_\mu = (-1)^\mu\mu\frac{1}{1-q^{2\mu}}= (-1)^\mu|\mu|\frac{q^{|\mu|-\mu}}{1-q^{2|\mu|}},\nonumber\\
(S_2)_\mu = (-1)^\mu\mu\frac{q^\mu}{1-q^{2\mu}}= (-1)^\mu|\mu|\frac{q^{|\mu|}}{1-q^{2|\mu|}},\nonumber\\
(S_3)_\mu = \mu\frac{q^\mu}{1-q^{2\mu}}= |\mu|\frac{q^{|\mu|}}{1-q^{2|\mu|}}\label{Snumu}
\end{gather}
(see Appendix~\ref{appwp} for derivations of these formulas). From this the following result is obtained by straightforward computations.

\begin{Lemma}\label{lemma2}
The functions in \eqref{Fn} satisfy
\begin{gather}\label{HFn}
\left( \frac{{\rm i}}{\pi}\kappa \frac{\partial}{\partial\tau} + H\big(x;\{g_\nu\}_{\nu=0}^3\big) \right)F_n(x) = \big( C_0 + E^{(0)}_n \big)F_n(x) -\sum_{\mu\in{\mathbb{Z}'}} S_\mu F_{n-\mu}(x)
\end{gather}
with $E^{(0)}_n$ in \eqref{PTsolution},
\begin{gather}\label{gammanu1}
S_\mu\equiv \sum_{\nu=0}^3\gamma_\nu(S_\nu)_\mu,\qquad \gamma_\nu \equiv \tilde g_\nu(\tilde g_\nu-1),
\end{gather}
$\tilde g_\nu\equiv \lambda-g_\nu$, and
\begin{gather}\label{C0}
C_0 = \kappa ^2\left(\frac1{12}-\frac{\eta_1}{\pi}\right) -\sum_{\nu=0}^3 g_\nu(g_\nu-1)\frac{\eta_1}{\pi}.
\end{gather}
\end{Lemma}
(The proof can be found at the end of this section.)

To proceed we make the ansatz
\begin{gather}\label{ansatz11}
\psi(x) = \mathcal{N} \sum_{m\in\mathbb{Z}}\alpha(m) F_m(x)
\end{gather}
with $\mathcal{N}$ a $q$-independent normalization constant to be determined. Inserting this in \eqref{Heun} and using Lemma~\ref{lemma2} we obtain
\begin{gather*}
\sum_{m\in\mathbb{Z}}\left(\left(\frac{{\rm i}}{\pi}\kappa \frac{\partial}{\partial\tau} + E^{(0)}_m -\mathcal{E} \right)\alpha(m)
-\sum_{\mu\in{\mathbb{Z}'}}S_\mu\alpha(m+\mu) \right) F_m(x)=0
\end{gather*}
with $\mathcal{E}\equiv E-C_0$. It follows that the function in \eqref{ansatz11} is a solution of \eqref{Heun} if the coef\/f\/icients~$\alpha(m)$ and $\mathcal{E}$ satisfy
\begin{gather}\label{aneqs1}
\left(\frac{{\rm i}}{\pi}\kappa \frac{\partial}{\partial\tau} + E^{(0)}_m -\mathcal{E} \right)\alpha(m)
= \sum_{\mu\in{\mathbb{Z}'}}S_\mu\alpha(m+\mu) .
\end{gather}
Inserting \eqref{Snumu} and \eqref{gammanu1} and changing variables from $\tau$ to $q={\rm e}^{{\rm i}\pi\tau}$ allow us to write this as
\begin{gather}
\left( - \kappa q\frac{\partial}{\partial q} +  E_{m}^{(0)} -\mathcal{E} \right) \alpha(m) = \sum_{\mu=1}^\infty \mu(\gamma_0 + (-1)^\mu\gamma_1) \alpha(m+\mu) \nonumber\\
\qquad{} + \sum_{\mu=1}^\infty \mu\left( \frac{[\gamma_0 + (-1)^\mu\gamma_1)]q^{2\mu}}{1-q^{2\mu}} +\frac{[(-1)^\mu\gamma_2+\gamma_3]q^{\mu}}{1-q^{2\mu}}\right) [ \alpha(m+\mu) + \alpha(m-\mu)]\label{aneqs2}
\end{gather}
equivalent to \eqref{aneqs}. This proves that $\psi(x)$ in \eqref{ansatz11} and $E=\mathcal{E}+C_0$ solve \eqref{Heun} provided \eqref{aneqs2} is fulf\/illed.

We are left to show that a solution $\alpha(m)=\alpha_n(m)$, $\mathcal{E}=\mathcal{E}_n$ satisfying the condition in \eqref{ic} is such that \eqref{eJacobi} holds true. This is done in Section~\ref{secTrigonometric}.

\begin{proof}[Proof of Lemma~\ref{lemma2}] It follows from \eqref{kernel} that
\begin{gather}\label{key1}
\left(\frac{{\rm i}}{\pi} \kappa \frac{\partial}{\partial\tau} + H\big(x;\{g_\nu\}_{\nu=0}^3\big) \right)G^{\kappa} \mathcal{K}(x,y) = \big( H\big(y;\{\tilde g_\nu\}_{\nu=0}^3\big) + C_{1,1}' \big) G^{\kappa} \mathcal{K}(x,y)
\end{gather}
with
\begin{gather}\label{C11p}
C_{1,1}' \equiv C_{1,1} + \kappa^2\left(\frac1{12}-\frac{\eta_1}{\pi}\right)
\end{gather}
(we computed
\begin{gather*}
\frac{{\rm i}}{\pi}\frac1{G}\frac{\partial}{\partial\tau} G = \sum_{n=1}^\infty \frac{2nq^{2n}}{1-q^{2n}} = \sum_{n=1}^\infty \sum_{\ell=1}^\infty 2nq^{2n\ell} =\sum_{\ell=1}^\infty \frac{2q^{2\ell}}{(1-q^{2\ell})^2}= \frac1{12}-\frac{\eta_1}{\pi}
\end{gather*}
using \eqref{G} and \eqref{eta1pi}). To compute the r.h.s.\ in \eqref{key1} we use
\begin{gather*}
H\big(y;\{\tilde g_\nu\}_{\nu=0}^3\big) = -\frac{\partial^2}{\partial y^2} - \sum_{\mu\in{\mathbb{Z}'}}S_\mu{\rm e}^{{\rm i} \mu y}-\sum_{\nu=0}^3\gamma_\nu\frac{\eta_1}{\pi},\qquad 0<\Im(y)<\pi\Im(\tau)
\end{gather*}
obtained by inserting \eqref{wpseries} in \eqref{H} and using \eqref{gammanu1}.
Using \eqref{genfun} and equating terms with the same factor ${\rm e}^{{\rm i} (n+\frac{1}{2}(g_0+g_1))y}$ give \eqref{HFn} and
\begin{gather*}
C_0 = C_{1,1}'-\sum_{\nu=0}^3\gamma_\nu\frac{\eta_1}{\pi}.
\end{gather*}
Using \eqref{C11} and \eqref{C11p}, inserting \eqref{gammanu1}, and recalling \eqref{lambda}, one obtains the formula in \eqref{C0}.
\end{proof}

\subsection{Trigonometric limit}\label{secTrigonometric}
For $q=0$ the equations in \eqref{aneqs}--\eqref{ic} simplify to
\begin{gather}\label{aneqs0}
\big[E^{(0)}_{m}-\mathcal{E}^{(0)}_n\big]\alpha^{(0)}_n(m) =\sum_{\mu=1}^{n-m}\mu \gamma_0^\mu \alpha^{(0)}_n(m+\mu)
\end{gather}
and $\alpha^{(0)}_n(n)=1$; we use the superscript ``$(0)$'' to indicate that a quantity is for $q=0$. This system of equations is an eigenvalue problem for a non-degenerate triangular matrix which can be solved by recursion: $\mathcal{E}_n^{(0)}=E^{(0)}_n$, $\alpha^{(0)}_n(n)=1$, and
\begin{gather}\label{Eq:an0iteration}
\alpha_n^{(0)}(m<n) = \frac{1}{b^{(0)}_n(m-n)}\sum_{\mu=1}^{n-m} \mu\gamma_0^\mu \alpha^{(0)}_n(m+\mu)
\end{gather}
with
\begin{gather}\label{bnm}
b_n^{(0)}(m) \equiv E^{(0)}_{m+n}-E^{(0)}_{n} = m(m+2n+g_0+g_1),
\end{gather}
provided that $g_0+g_1$ is such that
\begin{gather}\label{nores0}
b_n^{(0)}(-m)\neq 0 \qquad \forall\, m=1,2,\ldots,n.
\end{gather}
The latter condition is implied by our assumption that $-(g_0+g_1)\notin\mathbb{N}$.

We recall that $f^{(0)}_m(z)\equiv \left. f_m(z)\right|_{q=0}$ is zero for $m<0$ and a polynomial satisfying \eqref{f0n} for $m\geq 0$ (see Lemma~\ref{Lemma:fn2}), and thus
\begin{gather}\label{Pn0}
\mathcal{P}^{(0)}_n(z) \equiv \left. \mathcal{P}_n(z)\right|_{q=0} = {\mathcal N}_n \sum_{m=0}^n \alpha_n^{(0)}(m)f_{m}^{(0)}(z) = {\mathcal N}_n \binom{-\lambda}{n}(-2z)^n+O\big(z^{n-1}\big)
\end{gather}
is a polynomial of degree $n$. Moreover, the special case $q=0$ of results proved above implies that $\psi_n^{(0)}(x)\equiv (\sin \frac{1}{2} x)^{g_0}(\cos \frac{1}{2} x)^{g_1}\mathcal{P}_n^{(0)}(\cos x)$ is a solution of \eqref{PT} with $E^{(0)}=E^{(0)}_n$. Thus, by Lemma~\ref{lemmaUniqueness}, \eqref{eJacobi} holds true provided that the coef\/f\/icient of the leading term in~\eqref{Pn0} agrees with the one in \eqref{Pnormalization}. This f\/ixes the normalization constant $\mathcal{N}_n$ as in \eqref{cNn} and completes the proof of Proposition~\ref{prop0}.

\section{Recursive algorithms}\label{secPerturbative}
We now present algorithms to compute the expansion coef\/f\/icients $\mathcal{P}_n^{(\ell)}(z)$ and $\mathcal{E}_n^{(\ell)}$ of our solution def\/ined in \eqref{P}, which are based on Proposition~\ref{prop0}.
The f\/irst algorithm given in Section~\ref{secPerturbative1} is such that it can be used for all values of $\kappa$, including $\kappa=0$.
We then give a second algorithm, which is a variant of the f\/irst, which is simpler but requires $\kappa\neq 0$; see Section~\ref{secPerturbative2}.

\subsection{Algorithm I}\label{secPerturbative1}
We compute the functions $\mathcal{P}_n^{(\ell)}(z)$ and $\mathcal{E}_n^{(\ell)}$ def\/ined in \eqref{P} by solving the equations in \eqref{aneqs}--\eqref{ic} with the ansatz
\begin{gather}\label{anEnseries}
\alpha_n(m) = \sum_{\ell=0}^\infty \alpha_n^{(\ell)}(m)q^\ell,\qquad \mathcal{E}_n = \sum_{\ell=0}^\infty \mathcal{E}_n^{(\ell)}q^\ell,
\end{gather}
which leads to recursive relations allowing a straightforward solution. Inserting this solution in \eqref{Pnseries} and using Lemma~\ref{Lemma:fn2} one obtains representations of the functions $\mathcal{P}_n^{(\ell)}(z)$ which make manifest that they are polynomials of degree $n+\ell$ in $z$ for $n+\ell\geq 0$ and zero otherwise.

To formulate this result we f\/ind it useful to introduce the shorthand notation
\begin{gather}\label{bnml}
b_n^{(\ell)}(k) \equiv E^{(0)}_{n+k}-E^{(0)}_n-\kappa\ell=k(k+2n+g_0+g_1) -\kappa\ell
\end{gather}
($E^{(0)}_n$ in \eqref{PTsolution}) and recall the def\/inition of $\lambda$ in \eqref{lambda}. Note that the denominators in the fractions in \eqref{resonance1} are equal to $b_n^{(\ell)}(m-n)$.

\begin{Proposition}\label{prop1}
Let $n\in\mathbb{Z}$, $-\lambda\notin\mathbb{N}_0$ for $n>0$ and $-(g_0+g_1)\notin\mathbb{N}_0$ for $n<0$, $\{f^{(\ell)}_m(z)\}_{m\in\mathbb{Z}}$ the functions defined and characterized in Lemma~{\rm \ref{Lemma:fn2}}, and assume that either $\Im(\kappa)\neq 0$ and $g_0+g_1\in\mathbb{R}$ or $\kappa=0$ and $g_0+g_1\notin \mathbb{Z}$.\footnote{Note that the latter are the no-resonance conditions discussed in Section~\ref{secResonances}.} Then the non-stationary Heun equation in \eqref{Heun} has a~unique solution as in \eqref{solution}--\eqref{P} given by
\begin{gather}\label{Pnellz}
\mathcal{P}_n^{(\ell)}(z) = \mathcal{N}_n\sum_{\ell'=0}^{\ell} \sum_{m=-\ell'}^{n+\ell-\ell'} \alpha_n^{(\ell-\ell')}(m)f_m^{(\ell')}(z)
\end{gather}
with $\mathcal{N}_n$ in \eqref{cNn}, and $\alpha_n^{(\ell)}(m)$, $\mathcal{E}^{(\ell)}_n$ are determined by the following recursion relations,
\begin{gather}
 \alpha_n^{(\ell)}(m) = \frac1{b_n^{(\ell)}(m-n)}\Biggl( \sum_{\ell'=1}^\ell \mathcal{E}_n^{(\ell')}\alpha_n^{(\ell-\ell')}(m) + \sum_{\mu=1}^{n-m+\ell}\mu \gamma_0^\mu \alpha_n^{(\ell)}(m+\mu) \nonumber\\
 \hphantom{\alpha_n^{(\ell)}(m) =}{} + \sum_{\ell'=0}^{\ell-1} \sum_{\mu=1}^{\ell-\ell'}\sum_{k=1}^{\left \lfloor{\frac{\ell-\ell'}{\mu}}\right \rfloor } \mu \gamma_k^\mu \delta_{\ell,\ell'+k\mu} \big[\alpha_n^{(\ell')}(m+\mu)+\alpha_n^{(\ell')}(m-\mu)\big] \Biggr)\label{recur22}
\end{gather}
for $m\neq n$, $\ell\geq 0$, and
\begin{gather}\label{Eellneqs}
\mathcal{E}^{(\ell)}_n =
- \sum_{\mu=1}^{\ell}\mu \gamma_0^\mu \alpha_n^{(\ell)}(n+\mu) -
\sum_{\ell'=0}^{\ell-1} \sum_{\mu=1}^{\ell-\ell'}\sum_{k=1}^{\left \lfloor{\frac{\ell-\ell'}{\mu}} \right\rfloor } \mu\gamma_k^\mu \delta_{\ell,\ell'+k\mu} \big[\alpha_n^{(\ell')}(n+\mu)+\alpha_n^{(\ell')}(n-\mu)\big]
\end{gather}
for $\ell\geq 1$, together with the following conditions
\begin{gather}\label{ic3}
\alpha_n^{(0)}(n)=1,\qquad \alpha_n^{(\ell)}(n)=0\qquad \forall\, \ell\geq 1,\qquad \alpha_n^{(\ell)}(m) = 0\qquad \forall\, m>n+\ell,\quad \ell\geq 0.
\end{gather}
The functions $\mathcal{P}_n^{(\ell)}(z)$ thus obtained are polynomials of degree $n+\ell$ in $z$ for $n+\ell\geq 0$ and zero otherwise, and \eqref{eJacobi} holds true provided $-(g_0+g_1)\notin\mathbb{N}$.
\end{Proposition}
(The proof can be found at the end of this section.)

The equations in \eqref{recur22}--\eqref{ic3} comprise the f\/irst of our perturbative solution algorithms.
It is important to note that it has a triangular structure and is f\/inite, which implies that it determines each coef\/f\/icient $\alpha_n^{(\ell)}(m)$ and $\mathcal{E}_n^{(\ell)}$ as a sum of f\/initely many terms.
To be specif\/ic: The set of pairs $(\ell,m)$ is partially ordered as follows,
\begin{gather}
\label{order}
(\ell',m')\prec (\ell,m)\Leftrightarrow (\ell'<\ell) \qquad \text{or}  \qquad ( \ell'=\ell\mbox{ and } m'>m ) ,
\end{gather}
and \eqref{recur22} and \eqref{Eellneqs} have the form
\begin{gather*}
\alpha_n^{(\ell)}(m) = \sum_{(\ell',m')\prec (\ell,m)} B_{(\ell,m),(\ell',m')}\alpha_n^{(\ell')}(m'),\\
 \mathcal{E}^{(\ell)}_n=\sum_{(\ell',m')\prec(\ell,n)} C_{(\ell,n),(\ell',m')}\alpha_n^{(\ell')}(m')
\end{gather*}
with explicitly known coef\/f\/icients $B_{(\ell,m),(\ell',m')}$ and $C_{(\ell,n),(\ell',m')}$ which, for f\/ixed f\/irst argument, are non-zero only for {\em finitely} many values of the second argument $(\ell',m')$. For the convenience of the reader we compute the solution $\alpha_n^{(\ell)}(m)$, $\mathcal{E}_n^{(\ell)}$ of this system of equations analytically for $\ell=0,1,2$ in Appendix~\ref{appExplicitResults}, and we obtain the following result for the generalized eigenvalues, $\mathcal{E}_n=(P/2)^2 +\mathcal{E}_n^{(1)}q + \mathcal{E}_n^{(2)}q^2+O(q^3)$ with $P=2n+g_0+g_1$ and
\begin{gather}\label{cEn1}
\mathcal{E}_n^{(1)}= \gamma_0^1\gamma_1^1\left(\frac1{P-1}-\frac1{P+1-\kappa}\right),\\
\mathcal{E}_n^{(2)} = \left(\gamma_0^1\right)^2\left(\frac{1}{P-1} - \frac{1}{P+1-2\kappa} \right) + \left(\gamma_1^1\right)^2 \left(\frac{1}{P-1+\kappa} - \frac{1}{P+1-\kappa} \right)\nonumber\\
\hphantom{\mathcal{E}_n^{(2)}=}{} + 4\gamma_0^0\gamma_1^0 \left( \frac{1}{2(P-2)} - \frac{1}{2(P+2)-2\kappa}\right)
- 2 (\gamma_0^1)^{2} \gamma_1^0 \left( \frac{1}{[2(P-2)][P-1]} \right. \nonumber\\
\left.\hphantom{\mathcal{E}_n^{(2)}=}{}
 - \frac{1}{[P-1][P+1-2\kappa]} +\frac{1}{[P+1-2\kappa][2(P+2)-2\kappa]}\right)\nonumber\\
 \hphantom{\mathcal{E}_n^{(2)}=}{}
  - 2\gamma_0^0 \left(\gamma_1^1\right)^2 \left( \frac{1}{[2(P-2)][P-1+\kappa]} - \frac{1}{[P-1+\kappa][P+1-\kappa]} \right.\nonumber\\
\left.\hphantom{\mathcal{E}_n^{(2)}=}{}  + \frac{1}{[P+1-\kappa][2(P+2)-2\kappa]} \right)
- (\gamma_0^1)^2(\gamma_1^1)^2 \left( \frac{1}{[P-1+\kappa][P-1]^2} \right. \nonumber\\
\hphantom{\mathcal{E}_n^{(2)}=}{} - \frac{1}{[P+1-2\kappa][P+1-\kappa]^2}- \frac{1}{[P-1+\kappa][P-1][P+1-\kappa]}
\nonumber\\
\hphantom{\mathcal{E}_n^{(2)}=}{} + \frac{1}{[P-1][P+1-\kappa][P+1-2\kappa]}  -\frac{1}{[2(P-2)][P-1][P-1+\kappa]}\nonumber\\
\left.\hphantom{\mathcal{E}_n^{(2)}=}{}
+ \frac{1}{[P+1-\kappa][P+1-2\kappa][2(P+2)-2\kappa]} \right).\label{cEn2}
\end{gather}
This computation can be straightforwardly implemented in a symbolic programming language like MAPLE or MATHEMATICA and, in this way, extended to higher values of $\ell$.\footnote{We computed $\mathcal{E}_n^{(\ell)}$ up to $\ell=6$ using ordinary laptops and with computation times of the order of minutes.}

As proved by Ruijsenaars in \cite{RHeun}, the Heun equation in \eqref{Heun} for $\kappa=0$ has solutions $\psi_n(x)$ and corresponding eigenvalues $E_n$ such that the latter are invariant under all permutations of the following af\/f\/ine combinations of coupling parameters,
\begin{gather}\label{S4}
c_0\equiv g_0+g_2-1,\qquad c_1\equiv g_1+g_3-1,\qquad c_2\equiv g_1-g_3,\qquad c_3\equiv g_0-g_2.
\end{gather}
We used MAPLE to check that the coef\/f\/icients $\mathcal{E}_n^{(\ell)}$, $\ell=1,2,3,4,5$, determined by the algorithm in Proposition~\ref{prop2} have this permutation symmetry for $\kappa=0$ (but not for $\kappa\neq 0$).

\begin{proof}[Proof of Proposition~\ref{prop1}]
Insert the ansatz in \eqref{anEnseries} and the geometric series for $q^{\mu}/(1-q^{2\mu})$ into \eqref{aneqs}, compare terms which have the same power in $q$, and obtain
\begin{gather}
b_n^{(\ell)}(m-n) \alpha_n^{(\ell)}(m) = \sum_{\ell'=1}^\ell \mathcal{E}_n^{(\ell')}\alpha_n^{(\ell-\ell')}(m) + \sum_{\mu=1}^\infty\mu \gamma_0^\mu \alpha_n^{(\ell)}(m+\mu) \nonumber\\
\hphantom{b_n^{(\ell)}(m-n) \alpha_n^{(\ell)}(m) =}{}
+ \sum_{\ell'=0}^{\ell-1} \sum_{\mu=1}^\infty\sum_{k=1}^\infty \mu\gamma_k^\mu \delta_{\ell,\ell'+k\mu} \big[\alpha_n^{(\ell')}(m+\mu)+\alpha_n^{(\ell')}(m-\mu)\big],\label{recur}
\end{gather}
where $b_n^{(\ell)}(k)$ is short for $-\kappa\ell + E_{n+k}^{(0)} - \mathcal{E}_n^{(0)}$ for $\ell=0,1,2,\ldots$.
The f\/irst condition in \eqref{ic3} implies $\mathcal{E}_n^{(0)}=E^{(0)}_n$, and one obtains the formula for $b_n^{(\ell)}(k)$ in \eqref{bnml}.

One can check that the f\/irst two conditions in \eqref{ic3} are consistent with \eqref{recur} (details are given in Section~\ref{secTrigonometric}). Moreover, they imply that it is consistent to set
\begin{gather}\label{hwconditions1}
\alpha^{(\ell)}_n(m)=0\qquad \forall\, m>n+\ell,\quad  \ell\geq 0
\end{gather}
(this can be proved by induction: \eqref{hwconditions1} is true by assumption for $\ell=0$; for $\ell\geq 1$ it follows from~\eqref{recur} that $\alpha_n^{(\ell)}(m)$ for $m>n+\ell$ is a linear combination of $\alpha_n^{(\ell')}(m')$ with $\ell'<\ell$, $m'\geq m-\mu$ and $\mu\geq 1$ constrained by the Kronecker deltas in \eqref{recur}, i.e., $m'\geq m-(\ell-\ell')>n+\ell'$, which proves the claim). Thus one can restrict the second sum in~\eqref{recur} to $1\leq\mu\leq n-m+\ell$. One can also check that the inf\/inite sums in the terms in the second line in \eqref{recur} can be replaced by f\/inite ones: the Kronecker deltas there are non-zero for $\ell-\ell'=2k\mu$ and $\ell-\ell'=(2k-1)\mu$, and since $\mu,k\geq 1$ this is possible only if $1\leq\mu\leq\ell-\ell'$ and $1\leq k \leq (\ell-\ell')/(2\mu)+\frac{1}{2} $. Since $b_n^{(\ell)}(k)$ is non-zero for $m\neq n$ by assumption, we obtain~\eqref{recur22} from~\eqref{recur}. For $m=n$ and $\ell\geq 1$, the f\/irst and last conditions in \eqref{ic3} make~\eqref{recur} into an equation determining $\mathcal{E}^{(\ell)}_n$ as in~\eqref{Eellneqs}.
\end{proof}

\subsection{Algorithm II}\label{secPerturbative2}
For $\kappa\neq 0$, the condition $\alpha_n^{(\ell\geq 1)}(n)=0$ in \eqref{ic3} can be replaced by the condition $\mathcal{E}_n^{(\ell\geq 1)}=0$ and, in this way, a somewhat simpler solution algorithm is obtained.

 Recall the def\/initions of $\lambda$ in \eqref{lambda}$, \gamma_\mu^k$ in \eqref{gammanu} and $b_n^{(\ell)}(k)$ in \eqref{bnml}.

\begin{Proposition}\label{prop2}
Let $n\in\mathbb{Z}$, $-\lambda\notin\mathbb{N}_0$ for $n>0$ and $-(g_0+g_1)\notin\mathbb{N}_0$ for $n<0$, $\big\{f^{(\ell)}_m(z)\big\}_{m\in\mathbb{Z}}$ the functions defined and characterized in Lemma~{\rm \ref{Lemma:fn2}}, and assume that $\Im(\kappa)\neq 0$ and $g_0+g_1\in\mathbb{R}$.\footnote{The latter are no-resonance conditions discussed in Section~\ref{secResonances}.}
Then the non-stationary Heun equation in \eqref{Heun} has a unique solution $\psi_n(x)$, $E_n$ as in \mbox{\eqref{solution}--\eqref{P}} given by~\eqref{cNn}, \eqref{Pnellz} and
\begin{gather}\label{Ensimple}
\mathcal{E}_n = \left(n+\frac{g_0+g_1}{2} \right)^2,
\end{gather}
with coefficients $\alpha_n^{(\ell)}(m)$ determined by the following recursion relations,
\begin{gather}
\alpha_n^{(\ell)}(m) = \frac1{b_n^{(\ell)}(m-n)}\Biggl( \sum_{\mu=1}^{n-m+\ell}\mu \gamma_0^\mu \alpha_n^{(\ell)}(m+\mu) \nonumber\\
\hphantom{\alpha_n^{(\ell)}(m) =}{} +\sum_{\ell'=0}^{\ell-1} \sum_{\mu=1}^{\ell-\ell'}\sum_{k=1}^{\left \lfloor{\frac{\ell-\ell'}{\mu}}\right \rfloor } \mu \gamma_k^\mu \delta_{\ell,\ell'+k\mu} \big[\alpha_n^{(\ell')}(m+\mu)+\alpha_n^{(\ell')}(m-\mu)\big] \Biggr)\label{recur33}
\end{gather}
for all $m\in{\mathbb{Z}'}$ if $\ell=0$ and all $m\in\mathbb{Z}$ if $\ell\geq 1$, together with the following conditions
\begin{gather}\label{ic33}
\alpha_n^{(0)}(n)=1,\qquad \alpha_n^{(\ell)}(m)=0\qquad \forall\, m>n+\ell,\quad \ell\geq 0.
\end{gather}
The functions $\mathcal{P}_n^{(\ell)}(z)$ thus obtained are polynomials of degree $n+\ell$ in $z$ for $n+\ell\geq 0$ and zero otherwise.
\end{Proposition}

It follows from \eqref{anEnseries} that the f\/irst and last conditions in \eqref{ic3} are equivalent to $\alpha_n^{I}(n)=1$, which give rise to $\mathcal{E}_n^I=E^{(0)}_n+O(q)$, whereas the condition in \eqref{ic33} and $\mathcal{E}_n^{(\ell\geq 1)}=0$ corresponding to relaxing the former conditions to $\alpha_n^{II}(n) = 1+O(q)$ and instead requiring $\mathcal{E}_n^{II}=E^{(0)}_n$ (the superscripts $I$ and $II$ here are to distinguish the results from the algorithms in Proposition~\ref{prop1} and \ref{prop2}, respectively). From this and \eqref{invariance} we conclude that the results of the two algorithms are related as follows,
\begin{gather}\label{relprops}
\alpha_n^{I}(m) = \frac{\alpha_n^{II}(m) }{\alpha_n^{II}(n)},\qquad \mathcal{E}_n^I = E^{(0)}_n + \frac{1}{\alpha_n^{II}(n)}\kappa q\frac{\partial}{\partial q}\alpha_n^{II}(n).
\end{gather}
This shows that the algorithm in Proposition~\ref{prop2} is a re-summation of the one in Proposition~\ref{prop1}: it follows from \eqref{relprops} that, by computing $\alpha^{II}_n(n)$ up to some order $O(q^{N+1})$, one obtains a Pad\'e approximation of the generalized eigenvalues $\mathcal{E}_n^I$ as follows,
\begin{gather}\label{Pade}
 \mathcal{E}_n^I = E^{(0)}_n +\frac{\kappa\sum\limits_{\ell=1}^N \ell\alpha_n^{(II,\ell)}(n)q^\ell}{1+ \sum\limits_{\ell=1}^{N-1}\alpha_n^{(II,\ell)}(n)q^\ell} + O\big(q^{N+1}\big) .
\end{gather}
As an example we give
\begin{gather*}
\alpha_n^{(II,1)}(n) = \frac{\mathcal{E}_n^{(1)}}{\kappa}, \qquad \alpha_n^{(II,2)}(n) = \frac{\mathcal{E}_n^{(2)}}{2\kappa} + \frac{\big(\mathcal{E}_n^{(1)}\big)^2}{2\kappa^2},
\end{gather*}
with $\mathcal{E}_n^{(1)}$ and $\mathcal{E}_n^{(2)}$ in \eqref{cEn1} and \eqref{cEn2}, which is obtained by straightforward computations using the algorithm in Proposition~\ref{prop2}; note that, when this is inserted in \eqref{Pade} for $N=2$, the singularities at $\kappa=0$ in the numerator and denominator cancel. From our results above it is clear that this cancellation of singularities at $\kappa=0$ happens for arbitrary $N$ (this can be proved by the same argument we used to prove Lemma~\ref{lemResonance} in Section~\ref{secResonances}). This implies that the limit $\kappa\to 0$ is well-def\/ined (in the sense made precise in Lemma~\ref{lemResonance}), and we thus obtain from \eqref{Pade} an interesting representation of the eigenvalues of the $BC_1$ Inozemtsev model.

\begin{Remark}\label{remComplexity}
It is interesting to note that \eqref{Pade} provides a simpler representation of the Taylor coef\/f\/icients $\mathcal{E}_n^{(\ell)}$ of the eigenvalues of the Heun equation for larger values of $\ell$: we computed $\mathcal{E}_n^{(I,\ell)}$ and $\alpha_n^{(II,\ell)}(n)$ using MAPLE for $\ell=1,2,3,4,5$ and found that the former have $2$, $18$, $162$, $1776$, $21002$ distinct terms,\footnote{The numbers $2$ and $18$ in comparison with \eqref{cEn1} and \eqref{cEn2} explain what we mean by ``number of distinct terms''.} whereas the latter have $2$, $18$, $148$, $1298$, $11632$ distinct terms, respectively.
Moreover, as discussed in Appendix~\ref{appCombinatorics}, there exist explicit formulas for the coef\/f\/icients $\mathcal{E}_n^{(I,\ell)}$ and $\alpha_n^{(II,\ell)}(n)$, and the ones for the latter are much simpler than the ones for the former.
\end{Remark}

\section{Solution to all orders}\label{secAllOrders}
In this section we use results in \cite{ELeCS2} to obtain perturbative solutions of the equation in \eqref{Heun} to all orders, in generalization of the results in Section~\ref{secTrigonometric} for $q=0$.

As shown in Section~\ref{Proofprop0}, for all $n\in\mathbb{Z}$, \eqref{Heun} has solutions $\psi_n(x)$, $E_n$ as in \eqref{solution} and \eqref{Pnseries}--\eqref{cNn} provided that the coef\/f\/icients $\alpha_n(m)$ and the (shifted) eigenvalue $\mathcal{E}_n$ are (suitable) solutions of the dif\/ferential-dif\/ference equations in \eqref{aneqs}.
We observe that, by introducing the notation\footnote{As explained further below, the following def\/ines operators $\mathbb{A}$ and $\mathbb{B}$ on the vector space of sequences $\alpha=\{\alpha(m)\}_{m\in\mathbb{Z}}$.}
\begin{gather}\label{AbbBbb}
(\mathbb{A}\alpha)(m) \equiv \left(-\kappa q\frac{\partial}{\partial q} + E^{(0)}_m\right) \alpha(m),\qquad
(\mathbb{B}\alpha)(m) \equiv \sum_{\mu\in{\mathbb{Z}'} } S_\mu \alpha(m+\mu)
\end{gather}
with $E^{(0)}_m$ in \eqref{PTsolution} and $S_\mu$ in \eqref{gammanu1}, one can write \eqref{aneqs} in a simple way as follows,
\begin{gather}\label{eCSeq}
[\mathbb{A}-\mathcal{E} ]\alpha(m) = (\mathbb{B}\alpha)(m)
\end{gather}
(we f\/ind it convenient to suppress the dependence on $n$ here). We note that the notation introduced here is such that we can use results in \cite[Section~4.1]{ELeCS2} as they stand. We f\/irst consider the special case $\kappa=0$ (Heun equation) in Section~\ref{secAllOrders0}, and then the general case in Section~\ref{secAllOrders1}. We note that Theorem~\ref{Thm1} is a generalization of Proposition~\ref{prop1} restricted to $\kappa=0$; Theorem~\ref{Thm2} generalizes Proposition~\ref{prop2}.

\subsection{Heun equation}\label{secAllOrders0}
The Heun equation corresponds to the special case $\kappa=0$ of \eqref{Heun}. This is an eigenvalue equation for a 1D Schr\"odinger operator, and $\psi(x)$ and $E$ have a quantum mechanical interpretation as energy eigenfunction and eigenvalue, respectively. In this section we show how to use results in~\cite{ELeCS2} to obtain a perturbative solution for this case to all orders.

Using the Feshbach method and expanding a Neumann series yields implicit solutions of \eqref{AbbBbb}--\eqref{eCSeq} for $\kappa=0$ to all orders \cite{ELeCS2}. To formulate this result we introduce the following two functions of a complex variable $z$ ($n\in\mathbb{N}_0$ and $m\in\mathbb{Z}$ are f\/ixed),
 \begin{gather}
\Phi_n(z) \equiv -\sum_{s=2}^\infty \sum_{\mu_1\in{\mathbb{Z}'}} S_{\mu_1} \cdots \sum_{\mu_s\in{\mathbb{Z}'}} S_{\mu_s} \nonumber\\
\hphantom{\Phi_n(z) \equiv}{} \times
\frac{\delta(0,\mu_1+\cdots+\mu_s)}{\big[b^{(0)}_{n}(\mu_1)-z\big]_n\big[b^{(0)}_{n}(\mu_1+\mu_2)-z\big]_n\cdots \big[b^{(0)}_{n}(\mu_1+\cdots+\mu_{s-1})-z\big]_n}
\nonumber\\
\hphantom{\Phi_n(z)}{}
=  -\sum S_{\mu_1} S_{\mu_2} \frac{\delta(0,\mu_1+\mu_2)}{\big[b^{(0)}_{n}(\mu_1)-z\big]_n}\nonumber\\
\hphantom{\Phi_n(z) \equiv}{}
-\sum S_{\mu_1} S_{\mu_2} S_{\mu_3} \frac{\delta(0,\mu_1+\mu_2+\mu_3)}{\big[b^{(0)}_{n}(\mu_1)-z\big]_n\big[b^{(0)}_{n}(\mu_1+\mu_2)-z\big]_n} + \cdots\label{Phin}
\end{gather}
and
 \begin{gather}
G_n(z;m) \equiv \delta(m,n)+\sum_{s=1}^\infty
\sum_{\mu_1\in{\mathbb{Z}'}}S_{\mu_1} \cdots \sum_{\mu_s\in{\mathbb{Z}'}} S_{\mu_s}\nonumber \\
\hphantom{G_n(z;m) \equiv}{}
\times
\frac{\delta(m,n+\mu_1+\cdots+\mu_s)}{\big[b^{(0)}_{n}(m\!-\!n)\!-\!z\big]_n\big[b^{(0)}_{n}(m\!-\!n\!+\!\mu_1)\!-\!z\big]_n\cdots \big[b^{(0)}_{n}(m\!-\!n\!+\!\mu_1\!+\!\cdots\!+\!\mu_{s-1})\!-\!z\big]_n}
\nonumber\\
\hphantom{G_n(z;m)}{}
=  \delta(n,m) + \sum S_{\mu}\frac{\delta(n,m+\mu)}{\big[b^{(0)}_{n}(m-n)-z\big]_n} \nonumber\\
\hphantom{G_n(z;m) \equiv}{}
+ \sum S_{\mu_1} S_{\mu_2} \frac{\delta(n,m+\mu_1+\mu_2)}{\big[b^{(0)}_{n}(m-n)-z\big]_n\big[b^{(0)}_{n}(m-n+\mu_1)-z\big]_n} +\cdots\label{Gnm}
\end{gather}
with the Kronecker delta $\delta(n,m)$, $S_\mu$ in \eqref{gammanu1}, and the convenient shorthand notations
\begin{gather}\label{DEmn}
 \frac1{[b^{(0)}_{n}(k)-z]_n} \equiv \begin{cases} 1/\big[E^{(0)}_{n+k}-E^{(0)}_{n}-z\big], & k\neq 0, \\ 0, & k=0. \end{cases}
\end{gather}
It is important to note that these formulas should be interpreted in a perturbative sense as follows: for f\/ixed $N\in\mathbb{N}$, there are only f\/initely many terms in~\eqref{Phin} and~\eqref{Gnm} that are $O(q^\ell)$ with $\ell<N$ and thus, to obtain results up to $O(q^N)$-corrections, all inf\/inite series in our formulas can be truncated to {\em finite series}. The reason why it is convenient to write inf\/inite series is that it is dif\/f\/icult to give a simple general recipe for f\/inite-$N$ truncations; for example, in \eqref{Phin} one can restrict to $2\leq s\leq N$ and $-N\leq \mu_j\leq N$ for $j=1,2,\ldots,s$, but this is somewhat arbitrary since the resulting f\/inite sum still contains many terms which do not contribute to $O(q^N)$.

\begin{Theorem}\label{Thm1}
Let $n\in\mathbb{Z}$, $-\lambda\notin\mathbb{N}_0$ for $n>0$ and $-(g_0+g_1)\notin\mathbb{N}_0$ for $n<0$, $\big\{f^{(\ell)}_m(z)\big\}_{m\in\mathbb{Z}}$ the functions defined and characterized in Lemma~{\rm \ref{Lemma:fn2}}. Then the Heun equation in \eqref{Heun} for $\kappa=0$ and $g_0+g_1\notin\mathbb{Z}$ has a unique solution as in \eqref{solution} and \eqref{Pnseries}--\eqref{cNn} provided that
\begin{gather*}
\mathcal{E}_n = E^{(0)}_n+\tilde{\mathcal{E}}_n,\qquad \alpha_n(m)=G_n\big(\tilde{\mathcal{E}}_n;m\big)
\end{gather*}
with $\tilde{\mathcal{E}}_n$ the unique solution of the equation
\begin{gather*}
\tilde{\mathcal{E}}_n = \Phi_n\big(\tilde{\mathcal{E}}_n\big)
\end{gather*}
vanishing like $O(q)$ as $q\to 0$.
\end{Theorem}

(The proof can be found at the end of this section.)

From this result one can obtain the following fully explicit formula for the generalized eigenvalues of the non-stationary Heun equation,
\begin{gather}\label{Enexplicit}
\mathcal{E}_n = E^{(0)}_n + \sum_{m=1}^\infty \sum_{k_0,k_1,\ldots,k_{m-1}\in\mathbb{N}}\!\delta\left(\sum_{r=0}^{m-1}k_r, m\right)\!
\delta\left(\sum_{r=1}^{m-1}rk_r, m-1\right)\! (m-1)! \prod_{r=1}^m \frac{[\Phi_n^{(r)}]^{k_r}}{k_r!}\!\!\!
\end{gather}
with
 \begin{gather}
\Phi^{(r)}_n(z) \equiv -\sum_{s=2}^\infty  \sum_{\mu_1,\ldots,\mu_s\in{\mathbb{Z}'}}\sum_{r_1,\ldots,r_{s-1}\in\mathbb{N}_0} S_{\mu_1}\cdots S_{\mu_s} \nonumber\\
 \hphantom{\Phi^{(r)}_n(z) \equiv}{} \times
\frac{\delta(0,\mu_1+\cdots+\mu_s)\delta(r_1+\cdots r_{s-1},r)}{[b^{(0)}_{n}(\mu_1)]^{1+r_1}_n[b^{(0)}_{n}(\mu_1+\mu_2)]^{1+r_2}_n\cdots [b^{(0)}_{n}(\mu_1+\cdots+\mu_{s-1})]^{1+r_{s-1}}_n},\label{Phin111}
\end{gather}
and similarly for $\alpha_n(m)$; see Theorem~4.3.1 in \cite{ELeCS2} (one can check that the proof of the later theorem applies as it stands to the present case). As explained in Appendix~\ref{appCombinatorics}, this formula can be used to turn the computation of the generalized eigenvalues $\mathcal{E}_n$ of the non-stationary Heun equation into a combinatorial problem; see~\eqref{Enell11} {\em ff}.

\begin{Remark}
The elliptic generalizations of the Jacobi polynomials $P^{\big(g_0-\frac{1}{2},g_1-\frac{1}{2}\big)}_n(z)$ provided by Theorem~\ref{Thm1} can be formally written as
\begin{gather*}
\mathcal{P}_n(z) = {\mathcal N}_n\oint_{\mathcal{C}}\frac{d\xi}{2\pi{\rm i}\xi}\frac{\prod\limits_{\nu=0}^3\Theta_{\nu+1}(\xi)^{\tilde g_\nu}}{\Theta(z,\xi)^\lambda}\tilde{\mathcal{P}}_n(\xi)
\end{gather*}
with
\begin{gather*}
\tilde{\mathcal{P}}_n(\xi) = \xi^{-n} +\sum_{s=1}^\infty \prod_{\kappa=1}^s\left( \sum_{\mu_\kappa\in{\mathbb{Z}'}} \frac{ S_{\mu_\kappa}}{\Big[b_n^{(0)}(m-n+\sum\limits_{\ell=1}^{\kappa-1}\mu_{\ell})-\tilde{\mathcal{E}}_n \Big]_n} \right) \xi^{-n-\mu_1-\cdots-\mu_s}.
\end{gather*}
It is possible to identify $\tilde{\mathcal{P}}_n(\xi)$ with a singular eigenfunction of the Inozemtsev Hamiltonian $H(y;\{\tilde g_\nu\}_{\nu=0}^3)$ appearing in the generalized kernel function identity in Lemma~\ref{Lemma:kernel} and, in this way, extend the interpretation of the kernel function method given in \cite{ELsigma} to the present case.
\end{Remark}

\begin{proof}[Proof of Theorem~\ref{Thm1}] (We are brief since interested readers can f\/ind further details in \cite{ELeCS2}. In particular, it is explained in \cite{ELeCS2} why we can ignore questions of convergence in the argument below.)

Let $V$ be the vector space of sequences $\{\alpha(m)\}_{m\in\mathbb{Z}}$ and regard $\mathbb{A}$ and $\mathbb{B}$ in \eqref{AbbBbb} as linear operators $V\to V$. Def\/ine projections $\mathbb{P}$ on $V$ as follows,\footnote{We write $\mathbb{P}$ short for $\mathbb{P}_n$.}
\begin{gather}\label{Pbb}
(\mathbb{P}\alpha)(m) \equiv \delta(n,m)\alpha(m)\qquad \forall\, \alpha\in V
\end{gather}
and $\mathbb{P}^\perp\equiv I-\mathbb{P}$. Then $\mathbb{A}$ in \eqref{AbbBbb} commutes with $\mathbb{P}$, and the equation $\mathbb{P}\alpha_0=\alpha_0$ is solved by $(\alpha_0)(m)\equiv \delta(n,m)$.
Thus Lemma~4.1.1 in \cite{ELeCS2} implies that
\begin{gather}\label{Feshbach}
\alpha = \big[I+\big(\mathbb{A}-\mathcal{E}-\mathbb{P}^\perp\mathbb{B}\big)^{-1}\mathbb{P}^\perp\mathbb{B}\big]\alpha_0 , \qquad
\mathcal{E}\alpha_0 = \mathbb{A}\alpha_0 -\mathbb{P}\mathbb{B}\alpha
\end{gather}
is a solution of \eqref{eCSeq}. Expanding the resolvent in a Neumann series we obtain \cite{ELeCS2}
\begin{gather}\label{alphanmEn}
\alpha(m) = \sum_{s=0}^\infty \big( \big([\mathbb{A}-\mathcal{E}]^{-1}\mathbb{P}^\perp \mathbb{B}\big)^s\alpha_0\big)(m),\!\!\!\qquad \mathcal{E} = E^{(0)}_n-\sum_{s=0}^\infty \big(\mathbb{B}\big([\mathbb{A}-\mathcal{E}]^{-1}\mathbb{P}^\perp \mathbb{B}\big)^s\alpha_0\big)(n).\!\!\!\!\!
\end{gather}
The ansatz $\mathcal{E}=E^{(0)}_n+ \tilde{\mathcal{E}}$ implies
\begin{gather}\label{Abbinv}
\big([\mathbb{A}-\mathcal{E}]^{-1} \mathbb{P}^\perp \alpha\big)(m) = \frac1{\big[b^{(0)}_{n}(m-n)-\tilde{\mathcal{E}} \big]_n}\alpha(m)
\end{gather}
using the shorthand notation in \eqref{DEmn}. Using \eqref{Abbinv} to compute \eqref{alphanmEn} we obtain $\alpha(m)=G_n(\tilde{\mathcal{E}};m)$ and $\tilde{\mathcal{E}} = \Phi_n(\tilde{\mathcal{E}})$ with the functions def\/ined in \eqref{Phin}--\eqref{Gnm} \cite{ELeCS2}. The latter should be interpreted as non-linear equation whose solution $\tilde{\mathcal{E}}=\tilde{\mathcal{E}}_n$ vanishing like $O(q)$ determines the solution of the Heun equation we are interested in; see \cite{ELeCS2} for details.
\end{proof}

\subsection{Time dependent Heun equation}\label{secAllOrders1}
We now present generalizations of the results in the previous section to the non-stationary case $\kappa\neq 0$.

The argument to solve \eqref{AbbBbb}--\eqref{eCSeq} in the proof of Theorem~\ref{Thm1} can be generalized to $\kappa\neq 0$ if the following projection is used,
\begin{gather}\label{Pbb1}
(\mathbb{P}\alpha)^{(\ell)}(m) = \delta(\ell,0)\delta(n,m)\alpha^{(\ell)}(m)
\end{gather}
where $\alpha^{(\ell)}(m)$ are def\/ined as coef\/f\/icients of the formal power series in $q$ (see \eqref{anEnseries}); with that, the results in \eqref{Feshbach} hold true as they stand.
However, in the present case, the second of these equations is trivially solved by $\mathcal{E}=E^{(0)}_n$, and this implies a stronger result:

\begin{Theorem}\label{Thm2}
Let $n\in\mathbb{Z}$, $-\lambda\notin\mathbb{N}_0$ for $n>0$ and $-(g_0+g_1)\notin\mathbb{N}_0$ for $n<0$, $\big\{f^{(\ell)}_m(z)\big\}_{m\in\mathbb{Z}}$ the functions defined and characterized in Lemma~{\rm \ref{Lemma:fn2}}, and assume that $\Im(\kappa)\neq 0$ and $g_0+g_1\in\mathbb{R}$.\footnote{The latter are no-resonance conditions discussed in Section~\ref{secResonances}.}
Then the non-stationary Heun equation in \eqref{Heun} has a unique solution $\psi_n(x)$, $E_n$ as in \mbox{\eqref{solution}--\eqref{P}} given by \eqref{cNn}, \eqref{Pnellz}, \eqref{Ensimple} and the coefficients
\begin{gather}
\alpha^{(\ell)}_n(m) = \sum_{s=0}^\infty \sum_{k_1,\ldots,k_s\in\mathbb{N}_0} \sum_{\mu_1,\ldots,\mu_s\in{\mathbb{Z}'}} S_{\mu_1}(k_1)\cdots S_{\mu_s}(k_s)
\nonumber\\ \times
\frac{\delta(\ell,|\mu_1|k_1+\cdots+|\mu_s|k_s)\delta(n,m+\mu_1+\cdots+\mu_s)}{\big[b^{(\ell)}_{n}(m\!-\!n)\big] \big[b^{(\ell-|\mu_1|k_1)}_{n}(m\!-\!n\!+\!\mu_1)\big] \cdots \big[b^{(\ell-|\mu_1|k_1\!-\!\cdots\!-\!|\mu_{s-1}|k_{s-1})}_{n}(m\!-\!n\!+\!\mu_1\!+\!\cdots\!+\!\mu_{s-1})\big]} \nonumber\\ =
\delta(\ell,0)\delta(n,m) + \sum S_\mu(k) \frac{\delta(\ell,|\mu|k)\delta(n,m+\mu)}{[b_n^{(\ell)}(n-m)]} \nonumber\\
 + \sum S_{\mu_1}(k_1) S_{\mu_2}(k_2) \frac{\delta(\ell,|\mu_1|k_1+|\mu_2|k_2)\delta(n,m+\mu_1+\mu_2)}{\big[b_n^{(\ell)}(n-m)\big]\big[b_n^{(\ell-|\mu_1|k_1)}(n-m+\mu_1)\big]} + \cdots\label{Gnm1}
\end{gather}
using the shorthand notations
\begin{gather}\label{Sellmu1}
S_{\mu}(k) \equiv \begin{cases} \mu \gamma_0^\mu & \mbox{if $k=0$ and $\mu>0$}, \\
|\mu| \gamma_{k}^\mu & \mbox{if $k\in\mathbb{N}$}, \\
0 & \mbox{otherwise }
\end{cases}
\end{gather}
with $\gamma_k^\mu$ in \eqref{gammanu}, and
\begin{gather}\label{bnml1}
\frac1{[b^{(\ell)}_n(k)]}\equiv \begin{cases} 0 & \mbox{if } \ell=0 \mbox{ and } k=0, \\
1/b_n^{(\ell)}(k) & \mbox{otherwise} \end{cases}
\end{gather}
with $b_n^{(\ell)}(k)$ in \eqref{bnml}.
\end{Theorem}
(The proof can be found at the end of this section.)

It is not dif\/f\/icult to convince one-selves that the coef\/f\/icients in \eqref{Gnm1} are identical with the ones determined by the algorithm in Proposition~\ref{prop2}.
Thus, by setting $m=n$ and using the second identity in \eqref{relprops}, we obtain a formula to all orders for the generalized eigenvalues $\mathcal{E}_n$ of the time dependent Heun equation:
\begin{gather}\label{Enexplicit1}
\mathcal{E}_n = E^{(0)}_n + \frac{\sum\limits_{\ell =1}^\infty \kappa\ell\alpha_n^{(\ell)}(n)q^\ell}{1+\sum\limits_{\ell=1}^\infty \ell\alpha_n^{(\ell)}(n)q^\ell}
\end{gather}
with $\alpha_n(m)$ in \eqref{Gnm1} for $m=n$; the limit $\kappa\to 0$ of this formula is non-trivial but well-def\/ined. It would be interesting to f\/ind a re-summation which makes this manifest.

We thus obtained two explicit formulas for the eigenvalues of the Heun equation in \eqref{Heun} for $\kappa=0$: the formula in \eqref{Enexplicit}, and the limit $\kappa\to 0$ of the formula in \eqref{Enexplicit1}.
The former has the advantage that it is manifestly f\/inite for $\kappa=0$, whereas the latter requires a non-trivial limit. However, as explained in Appendix~\ref{appCombinatorics}, the latter formula is much simpler than the former.

\begin{proof}[Proof of Theorem~\ref{Thm2}]
Let $V$ be the vector space of all sequences $\alpha=\{\alpha^{(\ell)}(m)\}_{\ell\in\mathbb{N}_0,m\in \mathbb{Z}}$ identif\/ied with $\alpha=\{\alpha(m)\}_{m\in\mathbb{Z}}$ as in \eqref{anEnseries}.
Then $\mathbb{A}$ and $\mathbb{B}$ in \eqref{AbbBbb} can be written as linear operators $V\to V$ as follows,
\begin{gather}\label{AbbBbb1}
(\mathbb{A}\alpha)^{(\ell)}(m) = \big({-}\kappa\ell+E^{(0)}_m\big)\alpha^{(\ell)}(m),\qquad (\mathbb{B}\alpha)^{(\ell)}(m) = \sum_{\ell'=0}^\ell \sum_{\mu\in{\mathbb{Z}'}} S^{(\ell')}_\mu \alpha^{(\ell-\ell')}(m+\mu)
\end{gather}
with
\begin{gather}\label{Sellmu}
S_{\mu}^{(\ell)} \equiv \begin{cases} \mu \gamma_0^\mu & \mbox{if $\ell=0$, $\mu>0$}, \\
|\mu| \gamma_{\ell/|\mu|}^\mu & \mbox{if $\ell/|\mu|\in\mathbb{N}$}, \\
0 & \mbox{otherwise }
\end{cases}
\end{gather}
for $\mu\in{\mathbb{Z}'}$ and $\ell\in\mathbb{N}_0$ (the latter follows with \eqref{gammanu}, \eqref{Snumu}, and \eqref{gammanu1}).
This allows to rewrite~\eqref{aneqs} as in~\eqref{eCSeq}.

It is clear that the projection $\mathbb{P}$ in \eqref{Pbb1} commutes with $\mathbb{A}$ in \eqref{AbbBbb1}, and that the equation $\mathbb{P}\alpha_0=\alpha_0$ is solved by
\begin{gather}\label{alpha00}
(\alpha_0)^{(\ell)}(m)=\delta(\ell,0)\delta(n,m).
\end{gather}
Thus Lemma~4.1.1 in \cite{ELeCS2} implies that the solution $\alpha$ and $\mathcal{E}$ of \eqref{aneqs} satisf\/ies \eqref{Feshbach} with that $\alpha_0$ ($\mathbb{P}^\perp=I-\mathbb{P}$).

The def\/initions imply
\begin{gather*}
(\mathbb{P}\mathbb{B}\alpha)^{(\ell)}(m) = \delta(\ell,0)\delta(n,m)\sum_{\mu=1}^\infty \mu\gamma_0^{\mu}\alpha^{(0)}(n+\mu).
\end{gather*}
We showed already in Section~\ref{secTrigonometric} that the solution $\alpha$ of \eqref{aneqs} is such that $\alpha^{(0)}(m>n)=0$, and thus $\mathbb{P}\mathbb{B}\alpha=0$.
With that we f\/ind that the second equation in \eqref{Feshbach} is solved by $\mathcal{E}=\mathcal{E}^{(0)}_n$.\footnote{This result was already obtained by a dif\/ferent argument in the proof of Proposition~\ref{prop2}.}
Using this the f\/irst equation in \eqref{alphanmEn} simplif\/ies to
\begin{gather}\label{alphanmEn1}
\alpha(m) = \sum_{s=0}^\infty \big( \big(\big[\mathbb{A}-E^{(0)}_n\big]^{-1}\mathbb{P}^\perp \mathbb{B}\big)^s\alpha_0\big)(m)
\end{gather}

Inserting def\/initions we f\/ind
\begin{gather*}
\big(\big[\mathbb{A}-E^{(0)}_n\big]^{-1}\mathbb{P}^\perp\mathbb{B}\alpha\big)^{(\ell)}(m) = \frac1{\big[b^{(\ell)}_n(m-n)\big]}\sum_{\ell'=0}^\ell \sum_{\mu\in{\mathbb{Z}'}} S^{(\ell')}_\mu \alpha^{(\ell-\ell')}(m+\mu)
\end{gather*}
using the notation in \eqref{bnml1}. With that one can prove by induction that
\begin{gather*}
\big(\big(\big[\mathbb{A}-E^{(0)}_n\big]^{-1}\mathbb{P}^\perp\mathbb{B}\big)^s\alpha\big)^{(\ell)}(m) = \sum_{\ell_1=0}^{\ell}\sum_{\ell_2=0}^{\ell-\ell_1}\cdots \sum_{\ell_s=0}^{\ell-\ell_1-\cdots-\ell_{s-1}}
\sum_{\mu_1,\ldots,\mu_s\in{\mathbb{Z}'}} S^{(\ell_1)}_{\mu_1}\cdots S^{(\ell_s)}_{\mu_s} \\ \qquad{}\times
\frac1{\big[b^{(\ell)}_n(m-n)\big]\big[b^{(\ell-\ell_1)}_n(m-n+\mu_1)\big]\cdots \big[b^{(\ell-\ell_1-\cdots-\ell_{s-1})}_n(m-n+\mu_1+\cdots+\mu_{s-1})\big]} \\
\qquad{} \times
\alpha^{(\ell-\ell_1-\cdots-\ell_s)}(m+\mu_1+\cdots+\mu_s)
\end{gather*}
for all $s=1,2,\ldots$. Inserting this in \eqref{alphanmEn1} and using \eqref{alpha00} give
\begin{gather}\label{Gnm11}
\alpha^{(\ell)}_n(m) = \sum_{s=0}^\infty \sum_{\ell_1,\ldots,\ell_s\in\mathbb{N}_0} \sum_{\mu_1,\ldots,\mu_s\in{\mathbb{Z}'}} S^{(\ell_1)}_{\mu_1}\cdots S^{(\ell_s)}_{\mu_s}
\\ \hphantom{\alpha^{(\ell)}_n(m) =}{} \times
\frac{\delta(\ell,\ell_1+\cdots+\ell_s)\delta(n,m+\mu_1+\cdots+\mu_s)}{b^{(\ell)}_{n}(m-n)b^{(\ell-\ell_1)}_{n}(m-n+\mu_1)\cdots b^{(\ell-\ell_1-\cdots-\ell_{s-1})}_{n}(m-n+\mu_1+\cdots+\mu_{s-1})}\nonumber
\end{gather}
(to simplify notation we extended the $\ell_j$-summations to all non-negative integers, which is possible due to the f\/irst Kronecker delta in the sum; the term for $s=0$ is by def\/inition equal to the r.h.s.\ in~\eqref{alpha00}).

The def\/inition of $S_\mu^{(\ell)}$ in \eqref{Sellmu} suggests to change summation variables from $\ell$ to $k$ such that $\ell=|\mu| k$ (to reduce the number of terms in the formula which give zero contributions). Wri\-ting~$S_\mu(k)$ short for $S_\mu^{(|\mu|k)}$ we obtain the result in \eqref{Gnm1}--\eqref{Sellmu1}.
\end{proof}

\section{Final remarks}\label{secFinal}
We showed that a solution method based on kernel functions and developed to solve the elliptic Calogero--Sutherland (eCS) model \cite{ELeCS0,ELeCS2} generalizes to the non-stationary Heun equation in~\eqref{Heun}. This suggests to us that also other elliptic problems can be treated by this method; as a possible candidate we mention the non-stationary generalization of the eigenvalue problem for the $BC_{N}$ Inozemtsev model for $N\geq 2$ \cite{I}, which def\/ines a natural many-variable generalization of the non-stationary Heun equation:
\begin{gather}
\Biggl( \frac{{\rm i}}{\pi}\kappa \frac{\partial}{\partial\tau} + \sum_{j=1}^N\Biggl( -\frac{\partial^2}{\partial x_j^2}   +\sum_{\nu=0}^3 g_\nu(g_\nu-1) \wp(x_j+\omega_\nu) \Biggr)\nonumber\\
\qquad{} + 2\lambda(\lambda-1)\sum_{1\leq j<k\leq N}\left\{\wp(x_j-x_k)+\wp(x_j+x_k) \right\} \Biggr) \psi(x)=E\psi(x) ;\label{BCN}
\end{gather}
note that generalized kernel function identities for this problem and a discrete set of $\kappa$-values (depending on the other model parameters) are known \cite{LT}.

We also obtained results beyond previous results about the eCS model: the non-stationary Heun equation in \eqref{Heun} is invariant under the transformations in \eqref{symmetry}, and we found that, for $\kappa\neq 0$, this symmetry can be exploited to construct simpler representations of the generalized eigenvalues $E$ which are useful even in the case $\kappa=0$; see \eqref{Enexplicit1}.
We expect that a formula similar to the one in \eqref{Enexplicit1} for the eigenvalues of the eCS model can be found, and that this would be interesting since it might allow to better understand the relations between the solutions of the eCS model in \cite{ELeCS2} and the one by Nekrasov and Shatiashvili \cite{NS}.

As discussed in Appendix~\ref{appCombinatorics}, the explicit formulas for the solutions of the Heun equation in Theorem~\ref{Thm1} can be regarded as translations of the problem to solve this equation into a combinatorial problem. It is remarkable that dif\/ferent elliptic problems lead to the same combinatorial problem: we f\/ind that the combinatorial structure of the solution of the eCS model and the non-stationary Heun equation are the same, and model details only af\/fect the basic building blocks of the solutions. We expect the same is true for other elliptic problems like the one in \eqref{BCN}. We believe that a similar remark applies to non-stationary extensions of elliptic problems. It would be interesting to explore these combinatorial aspects of our solution in future work.

One important question about the non-stationary Heun equation is uniqueness: which conditions determine a unique solution?
Our results shed light on this question: we f\/ind that the conditions in \eqref{solution}--\eqref{eJacobi} do not f\/ix the solution uniquely, and our results suggest the following further requirements:
\begin{gather}\label{condlast}
\mathcal{E}_n^{(\ell)}=0,\qquad \mathcal{P}_n^{(\ell)}(z)=O\big(z^{n+\ell}\big) \qquad \forall\, n+\ell\geq 0,\; \ell\geq 1.
\end{gather}
It would be interesting to prove that the conditions \eqref{solution}--\eqref{eJacobi} and \eqref{condlast} imply uniqueness.

We f\/inally note that kernel functions have been used since a long time to transform the Heun equation into integral equations \cite{Erd,LW}; see also \cite{LaSl,Novikov}.  This has provided powerful tools to study analytical properties of solutions; for example, this was used by Ruijsenaars in his work on the hidden permutation symmetry mentioned after \eqref{S4} \cite{RHeun}. Our approach is similar in that the kernel function we use determines the analytic structure of the solutions we construct.
However, there are also important dif\/ferences. For example, our kernel functions are {\em not} given by hypergeometric functions as the ones in \cite{Erd,LW}, and they are singular and {\em not} $L^2$ as the ones used in \cite{RHeun}. Moreover, our emphasis is on constructing explicit representations of solutions.

\appendix

\section{Special functions}\label{appSpecialFunctions}
For the convenience of the reader we collect def\/initions and properties of standard special functions that we use.

\subsection{Elliptic functions}\label{appA1}
The Weierstrass $\wp$-function with periods $(2\omega_1, 2\omega_3)$ is def\/ined as follows:
\begin{gather}\label{wpdef}
 \wp (x\,|\,\omega_1,\omega_3)\equiv \frac{1}{x^2}+ \sum_{(m,n)\in {\mathbb{Z}} \times {\mathbb{Z}} \setminus \{ (0,0)\} } \left\{ \frac{1}{(x-\Omega_{m,n})^2}-\frac{1}{\Omega_{m,n}^2}\right\}
\end{gather}
with $\Omega_{m,n}\equiv 2m\omega_1 +2n\omega_3$. We also recall the def\/inition of the theta functions,
\begin{gather}
\theta _1(x) \equiv 2\sum _{n=0}^{\infty } (-1)^{n} q^{(n+1/2)^2} \sin (2n+1)x,\nonumber \\
\theta _2(x) \equiv 2\sum _{n=0}^{\infty } q^{(n+1/2)^2} \cos (2n+1)x ,\nonumber\\
\theta _3 (x) \equiv 1+ 2\sum _{n=1}^{\infty } q^{n^2} \cos 2n x ,\nonumber\\
\theta _4 (x) \equiv 1+ 2\sum _{n=1}^{\infty } (-1)^{n} q^{n^2} \cos 2n x\label{thetanu}
\end{gather}
with
\begin{gather}\label{q}
q={\rm e}^{{\rm i}\pi\tau},\qquad \tau=\frac{\omega_3}{\omega_1}.
\end{gather}

We need the identity
\begin{gather}\label{wp22}
\wp(z|\pi,\pi\tau) = -\frac{\eta_1}{\pi} +\sum_{n\in\mathbb{Z}} \frac{1}{4\sin^2\big(\frac{1}{2}[z+2n\pi\tau]\big)}
\end{gather}
with
\begin{gather}\label{eta1pi}
\frac{\eta_1}{\pi} = \frac1{12}-\sum_{n=1}^\infty\frac{2q^{2n}}{(1-q^{2n})^2}
\end{gather}
(see \cite[23.8.3 and 23.8.5]{Dig10}).

Let
\begin{gather}\label{G}
G\equiv \prod_{n=1}^\infty\big(1-q^{2n}\big).
\end{gather}
From the product representations of the theta functions (see \cite[20.5.1--20.5.4]{Dig10}) we obtain
\begin{gather}
\theta_1\big(\tfrac{1}{2} y\big) = q^{1/4} {\rm e}^{{\rm i}\pi/2} {\rm e}^{-{\rm i} y/2}G\Theta_1\big({\rm e}^{{\rm i} y}\big),\nonumber \\
\theta_2\big(\tfrac{1}{2} y\big) = q^{1/4} {\rm e}^{-{\rm i} y/2}G\Theta_2\big({\rm e}^{{\rm i} y}\big),\nonumber\\
\theta_3\big(\tfrac{1}{2} y\big) = G\Theta_3\big({\rm e}^{{\rm i} y}\big),\nonumber\\
\theta_4\big(\tfrac{1}{2} y\big) = G\Theta_4\big({\rm e}^{{\rm i} y}\big)\label{tettet}
\end{gather}
and
\begin{gather}\label{tettet1}
\theta_1\big(\tfrac{1}{2}(x+y)\big)\theta_1\big(\tfrac{1}{2}(x-y)\big)= q^{1/2}G^2{\rm e}^{-{\rm i} y}\Theta\big(\cos(x),{\rm e}^{{\rm i} y}\big)
\end{gather}
with the functions $\Theta_\nu(\xi)$, $\nu=1,2,3,4$, and $\Theta(z,\xi)$ def\/ined in~\eqref{Thetanu} and~\eqref{Theta}.

\subsection{Jacobi polynomials}
We use the following def\/initions of the Pochhammer symbol
\begin{gather}\label{Pochhammer}
(x)_n\equiv x(x+1)\cdots (x+n-1)
\end{gather}
and binomial coef\/f\/icients
\begin{gather}\label{binomial}
\binom{x}{n}\equiv \frac{(x-n+1)_n}{n!}
\end{gather}
for complex $x$ and $n\in\mathbb{N}_0$. We will also use the extensions of these def\/initions to all $n\in\mathbb{Z}$ provided by the $\Gamma$ function, i.e., $(x)_n=\Gamma(x+n)/\Gamma(x)$ and similarly for the binomial coef\/f\/icients.

The Jacobi polynomials are given by
\begin{gather}\label{Series}
P_n^{(\alpha,\beta)}(z) \equiv \sum_{\ell =0}^n\frac{(n+\alpha+\beta+1)_\ell(\alpha+\ell+1)_{n-\ell}}{\ell!(n-\ell)!}\left(\frac{z-1}{2}\right)^\ell
\end{gather}
for $\alpha,\beta\geq -1$, $n=0,1,2,\ldots$, and complex $z$ (see \cite[18.5.7]{Dig10}). Thus
\begin{gather}\label{JacobiNormalization}
P_n^{(\alpha,\beta)}(z) = \frac{(n+\alpha+\beta+1)_n}{2^n n!} z^n + O\big(z^{n-1}\big) .
\end{gather}
We recall the def\/inition of the Gegenbauer (or ultraspherical) polynomials by their well-known generating function:
\begin{gather}\label{Gegenbauer}
\big(1-2z\xi+\xi^2\big)^{-\lambda} = \sum_{n=0}^\infty C^{(\lambda)}_n(z)\xi^n = \sum_{n=0}^\infty \frac{(2\lambda)_n}{\left(\lambda+\frac{1}{2}\right)_n} P_n^{(\lambda-\frac{1}{2},\lambda-\frac{1}{2})}(z)\xi^n
\end{gather}
(see \cite[18.12.4]{Dig10}). This is equivalent to the f\/irst identity in \eqref{Simple1} (to see this, note that $(n+2g)_n/[4^n(g)_n] = (g+\frac{1}{2})_n/(2g)_n$).

\section{Scaling}\label{appScaling}
In the literature dif\/ferent conventions for $\omega_1$ are used. We therefore state how our results for $\omega_1=\pi$ can be transformed into results for other values of $\omega_1$.

For arbitrary $\omega_1\neq 0$, the non-stationary Heun equation can be written as
\begin{gather}\label{Heun2}
\left( \frac{{\rm i}\pi}{\omega_1^2}\kappa \frac{\partial}{\partial\tau} -\frac{\partial^2}{\partial x^2} +\sum_{\nu=0}^3 g_\nu(g_\nu-1) \wp(x+\omega_\nu|\omega_1,\omega_3)\right)\psi(x)=E\psi(x)
\end{gather}
with $\omega_0=0$, $\omega_2=-\omega_1-\omega_1\tau$, $\omega_3=\omega_1\tau$ and $\Im(\tau)>0$. Let
\begin{gather*}
\psi(x)\equiv \psi\big(x,\tau;\omega_1,\{g_\nu\}_{\nu=0}^3,\kappa\big),\qquad E\equiv E\big(\tau;\omega_1,\{g_\nu\}_{\nu=0}^3,\kappa\big)
\end{gather*}
be a solution of this equation for generic values of $\omega_1$. Then this solution can be obtained from a corresponding one for $\omega_1=\pi$ as follows,
\begin{gather*}
\psi\big(x,\tau;\omega_1,\{g_\nu\}_{\nu=0}^3,\kappa\big) = \left(\frac{\pi}{\omega_1}\right)^{1/2} \psi\left(\pi x/\omega_1,\tau;\pi, \{g_\nu\}_{\nu=0}^3,\kappa\right),\\
E\big(\tau;\omega_1,\{g_\nu\}_{\nu=0}^3,\kappa\big) = \left(\frac{\pi}{\omega_1}\right)^2 E\left(\tau;\pi, \{g_\nu\}_{\nu=0}^3,\kappa\right)
\end{gather*}
(to see this, use $\wp(x|\omega_1,\omega_3)=s^2\wp(sx|s\omega_1,s\omega_3)$ to show that \eqref{Heun2} is invariant under the scaling transformation
\begin{gather*}
\big(x,\tau,\omega_1,\kappa,\{g_\nu\}_{\nu=0}^3,E\big)\to \big(sx,\tau,s\omega_1,\kappa,\{g_\nu\}_{\nu=0}^3,E/s^2\big),
\end{gather*}
and set $s=\pi/\omega_1$; we scale $\psi(x)$ such that its $L^2$-norm is invariant).

\section{Explicit solution to low orders for Algorithm I}
\label{appExplicitResults}
To shorten notation we present our results using the following notation
\begin{gather*}
\mathcal{E}_n^{(\ell)}=\mathcal{E}^{(\ell)},\qquad a^{(\ell)}(k)= \alpha_n^{(\ell)}(n+k)
\end{gather*}
using that
\begin{gather*}
b_n^{(\ell)}(k) = k(k+P)-\kappa\ell \equiv b^{(\ell)}(k)
\end{gather*}
only depends on $P$ in \eqref{Pdef}; see \eqref{bnml} and Remark~\ref{remP} for further explanations.

We give explicit formulas for the perturbative solution described in the main text for $\ell=0,1,2$.
Note that in the computations the conditions in \eqref{ic3} are used without this being mentioned.

\subsubsection*{Zeroth order} For $\ell=0$ we get (here and in the following, $k=m-n$),
\begin{gather}\label{recur0}
b^{(0)}(k)a^{(0)}(k) = \sum_{\mu=1}^{-k}\mu \gamma_0^\mu a^{(0)}(k+\mu)
\end{gather}
for $k=-1,-2,-3,\ldots$ and $a(0)=1$, which has the solution
\begin{gather*}
a^{(0)}(-1) = \frac{\gamma_0^1}{b^{(0)}(-1)} ,\\
a^{(0)}(-2) = \frac1{b^{(0)}(-2)} \big( \gamma_0^1a^{(0)}(-1) + 2\gamma_0^0 \big) =
 \frac{1}{b^{(0)}(-2)}\left(2\gamma_0^0 +\frac{(\gamma_0^1)^2}{b^{(0)}(-1)}\right),
\\
a^{(0)}(-3) = \frac1{b^{(0)}(-3)} \big( \gamma_0^1a^{(0)}(-2) + 2\gamma_0^0 a^{(0)}(-1) + 3\gamma_0^1 \big) \\
\hphantom{a^{(0)}(-3)}{} =\frac{\gamma_0^1}{b^{(0)}(-3)}\left(
 3+ 2 \gamma_0^0\left(\frac{1}{b^{(0)}(-2)}+\frac{1}{b^{(0)}(-1)}\right)+ \frac{(\gamma_0^1)^2}{b^{(0)}(-2)b^{(0)}(-1)} \right)
\end{gather*} etc.

\subsubsection*{First order}
For $\ell=1$ we get
\begin{gather}\label{recur1}
b^{(1)}(k) a^{(1)} (k) = \mathcal{E}^{(1)} a^{(0)}(k) +\sum_{\mu=1}^{1-k}\mu \gamma_0^\mu a^{(1)}(k+\mu) +
\gamma_1^1\big[a^{(0)}(k+1)+a^{(0)}(k-1)\big]
\end{gather}
for $k=1,0,-1,\ldots$, which has the solution
\begin{gather*}
 a^{(1)}(1) = \frac{\gamma_1^1}{b^{(1)}(1)},\\
\mathcal{E}^{(1)} = -\gamma_0^1 a^{(1)}(1) -\gamma_1^1 a^{(0)}(-1) = -\gamma_0^1\gamma_1^1\left(\frac1{b^{(0)}(-1)}+ \frac1{b^{(1)}(1)}\right)
\end{gather*}
implying the result in \eqref{cEn1},
\begin{gather*}
 a^{(1)}(-1) =\frac1{b^{(1)}(-1)}\big(\mathcal{E}^{(1)} a^{(0)}(-1) + 2\gamma_0^0 a^{(1)}(1) + \gamma_1^1\big(1+ a^{(0)}(-2)\big) \big)\\
 =  \frac{\gamma_1^1}{b^{(1)}(-1)}\left( 1+ 2\gamma_0^0\left( \frac1{b^{(0)}(-2)} + \frac1{b^{(1)}(1)} \right) +
\frac{(\gamma_0^1)^2}{b^{(0)}(-1)} \left( \frac1{b^{(0)}(-2)} -\frac1{b^{(0)}(-1)} -\frac1{b^{(1)}(1)}
 \right) \right)
\end{gather*}
etc.
\subsubsection*{Second order}
For $\ell=2$ we get
\begin{gather}
b^{(2)}(k) a^{(2)}(k) = \mathcal{E}^{(2)}a^{(0)}(k) + \mathcal{E}^{(1)}a^{(1)}(k) +\sum_{\mu=1}^{2-k}\mu \gamma_0^\mu a^{(2)}(k+\mu) +
\gamma_0^1\big[a^{(0)}(k+1)\nonumber\\
\qquad{}  +a^{(0)}(k-1)\big] + \gamma_1^1\big[a^{(1)}(k+1)+a^{(1)}(k-1)\big] +2\gamma_1^0\big[a^{(0)}(k+2)+a^{(0)}(k-2)\big]\label{recur2}
\end{gather}
for $k=2,1,0,\ldots$, which has the following solution
\begin{gather*}
 a^{(2)}(2) = \frac{1}{b^{(2)}(2)}\big( \gamma_1^1 a^{(1)}(1) + 2\gamma_1^0 \big) = \frac1{b^{(2)}(2)}\left( 2\gamma_1^0 + \frac{(\gamma_1^1)^2}{b^{(2)}(1)}\right) ,\\
 a^{(2)}(1) = \frac1{b^{(2)}(1)}\big(\mathcal{E}^{(1)} a^{(1)}(1) + \gamma_0^1 a^{(2)}(2) + \gamma_0^1+2\gamma_1^0 a^{(0)}(-1)\big)
 = \frac{ \gamma_0^1}{b^{(2)}(1)}\Biggl(1 + 2\gamma_1^0\\
 \hphantom{a^{(2)}(1) =}{} \times
 \left(\frac1{b^{(0)}(-1)}+\frac1{b^{(2)}(2)} \right) + (\gamma_1^1)^2\left(\frac1{b^{(2)}(2)b^{(1)}(1)} - \frac1{b^{(1)}(1)^2}-\frac1{b^{(1)}(1)b^{(0)}(-1)} \right) \Biggr) ,
\end{gather*}
and
\begin{gather*}
\mathcal{E}^{(2)} = -\gamma_0^1 a^{(2)}(1) -2\gamma_0^0 a^{(2)}(2) - \gamma_0^1 a^{(0)}(-1)
- \gamma_1^1\big[ a^{(1)}(1)+ a^{(1)}(-1)\big]-2\gamma_1^0 a^{(0)}(-2) \\
\hphantom{\mathcal{E}^{(2)}}{} =
-\left(\gamma_0^1\right)^2\left(\frac{1}{b^{(0)}(-1)} + \frac{1}{b^{(2)}(1)} \right) - \left(\gamma_1^1\right)^2 \left(\frac{1}{b^{(1)}(-1)} + \frac{1}{b^{(1)}(1)} \right) \\
\hphantom{\mathcal{E}^{(2)}=}{}
- 4\gamma_0^0\gamma_1^0 \left( \frac{1}{b^{(0)}(-2)} + \frac{1}{b^{(2)}(2)}\right) \\
\hphantom{\mathcal{E}^{(2)}=}{}
-  2 (\gamma_0^1)^{2} \gamma_1^0 \left( \frac{1}{b^{(0)}(-2) b^{(0)}(-1)} + \frac{1}{b^{(0)}(-1) b^{(2)}(1)} + \frac{1}{b^{(2)}(1)b^{(2)}(2)}\right) \\
\hphantom{\mathcal{E}^{(2)}=}{} -  2\gamma_0^0 \left(\gamma_1^1\right)^2 \left( \frac{1}{b^{(0)}(-2)b^{(1)}(-1)} + \frac{1}{b^{(1)}(-1)b^{(1)}(1)} + \frac{1}{b^{(1)}(1) b^{(2)}(2)} \right) \\
\hphantom{\mathcal{E}^{(2)}=}{}
+  (\gamma_0^1)^2(\gamma_1^1)^2 \left( \frac{1}{b^{(1)}(-1)b^{(0)}(-1)^{2}} + \frac{1}{b^{(2)}(1) b^{(1)}(1)^{2}} + \frac{1}{b^{(1)}(-1) b^{(0)}(-1) b^{(1)}(1)} \right. \\
\hphantom{\mathcal{E}^{(2)}=}{}
+  \left. \frac{1}{b^{(0)}(-1)b^{(1)}(1) b^{(2)}(1)} - \frac{1}{b^{(0)}(-2)b^{(0)}(-1)b^{(1)}(-1)} - \frac{1}{b^{(1)}(1)b^{(2)}(1)b^{(2)}(2)} \right)
\end{gather*}
equal to the result in \eqref{cEn2}.

\begin{Remark}
The results for the Heun equation ($\kappa=0$) are obtained from this by replacing
\begin{gather*}
b^{(\ell)}(k)\to b^{(0)}(k) \qquad \forall\, k,\; \ell.
\end{gather*}
For this case we found empirically that there is a useful symmetry as follows: one can obtain the formula for $a^{(\ell)}(k)$, $k=1,2,\ldots,\ell$, from the one for $a^{(\ell-k)}(-k)$ by the following transformation:
\begin{gather*}
b^{(0)}(k')\leftrightarrow b^{(0)}(-k'),\qquad \gamma_0^\mu\leftrightarrow \gamma_1^\mu
\end{gather*}
(i.e., swap the sign of the argument in all $b^{(0)}(k')$'s, and exchange all $\gamma_0^\mu$ and $\gamma_1^\mu$), and $\mathcal{E}^{(\ell)}$ is invariant under this transformation $\forall \ell$.
This allows to reduce the computations considerably.
\end{Remark}

\section{Details and proofs}
\label{appComputations}
We collect computational details and technical proofs.

\subsection{Proof of Lemma~\ref{Lemma:fn2}}
\label{secProofLemfn2}
We f\/irst consider the case $q=0$. We insert the def\/initions in \eqref{Thetanu}--\eqref{Theta} and use binomial series to expand
\begin{gather}
\sum_{n\in\mathbb{Z}}f_n^{(0)}(z)\xi^n = \frac{(1-\xi)^{\tilde{g}_0}(1+\xi)^{\tilde{g}_1}}{(1-2\xi z+\xi^2)^\lambda} =
\sum_{\nu_0,\nu_1,\nu_2=0}^\infty \binom{\tilde{g}_0}{\nu_0}\binom{\tilde{g}_1}{\nu_1}\binom{-\lambda}{\nu_2}(-\xi)^{\nu_0}\xi^{\nu_1}\big(\xi^2-2\xi z\big)^{\nu_2} \nonumber\\
\hphantom{\sum_{n\in\mathbb{Z}}f_n^{(0)}(z)\xi^n}{}
= \sum_{\nu_0,\nu_1,\nu_2=0}^\infty\sum_{m=0}^{\nu_2} \binom{\tilde{g}_0}{\nu_0}\binom{\tilde{g}_1}{\nu_1}\binom{-\lambda}{\nu_2}\binom{\nu_2}{m}(-1)^{\nu_0}\xi^{\nu_0+\nu_1+2\nu_2-m}(-2z)^{m}.\label{Eq:fn0}
\end{gather}
Comparing both sides we obtain
\begin{gather}\label{Eq:fn0explicit}
f_n^{(0)}(z) = \sum_{\nu_1,\nu_2,m} \binom{\tilde{g}_0}{n+m-\nu_1-2\nu_2}\binom{\tilde{g}_1}{\nu_1}\binom{-\lambda}{\nu_2}\binom{\nu_2}{m}(-1)^{\nu_0}(-2z)^{m}
\end{gather}
where the sum is over all non-negative integer triples $(\nu_1,\nu_2,m)$ such that $n+m-\nu_1-2\nu_2\geq 0$ and $m\leq \nu_2$.
From this it is obvious that $f_n^{(0)}(z)$ is a polynomial of degree $n$ in $z$ satisfying \eqref{f0n} (the highest power of $z$ is obtained for $(\nu_1,\nu_2,m)=(0,n,n)$).

We now consider the case $q\neq 0$. We rewrite \eqref{fn} using the result for $q=0$ above:
\begin{gather*}
\sum_{n\in\mathbb{Z}}f_n(z)\xi^n = \sum_{m=0}^\infty f_m^{(0)}(z)\xi^m\prod_{k=1}^\infty \frac{ \prod\limits_{\epsilon=\pm}\big[\big(1-\epsilon q^k\xi\big)\big(1-\epsilon q^k\xi^{-1}\big)\big]^{\tilde{g}_{k,\epsilon}}}{\big[\big(1-2q^{2k}\xi z+q^{4k}\xi^2\big)\big(1-2q^{2k}\xi^{-1}z+q^{4k}\xi^{-2}\big)\big]^\lambda}
\end{gather*}
with $(\tilde{g}_{k,+},\tilde{g}_{k,-})$ equal to $(\tilde{g}_0,\tilde{g}_1)$ and $(\tilde{g}_3,\tilde{g}_2)$ for even- and odd $k$, respectively. Expanding all factors in the product in binomial series the r.h.s.\ of this becomes a linear combination of terms
\begin{gather*}
f_m^{(0)}(z)\xi^m \prod_{k=1}^\infty \big(q^k\xi\big)^{\nu_{k,1}}\big(q^k\xi^{-1}\big)^{\nu_{k,2}}\big[q^{2k}\xi\big(q^{2k}\xi - z\big)\big]^{\nu_{k,3}} \big[q^{2k}\xi^{-1}\big(q^{2k}\xi^{-1} - z\big)\big]^{\nu_{k,4}}
\end{gather*}
with non-negative integers $m$ and $\nu_{k,j}$, which can be further expanded to a linear combination of terms
\begin{gather*}
f_m^{(0)}(z)\xi^m \prod_{k=1}^\infty \big(q^k\xi\big)^{\nu_{k,1}}\big(q^k\xi^{-1}\big)^{\nu_{k,2}}\big(q^{2k}\xi\big)^{2\nu_{k,3}-m_{k,3}}z^{m_{k,3}}\big(q^{2k}\xi^{-1}\big)^{2\nu_{k,4} -m_{k,4}}z^{m_{k,4}} \\
\qquad{} \equiv f_m^{(0)}(z) q^L z^{M}\xi^{m+N}
\end{gather*}
with non-negative integers $m_{k,3}\leq \nu_{k,3}$, $m_{k,4}\leq \nu_{k,4}$ and
\begin{gather*}
L=\sum_{k=1}^\infty k(\nu_{k,1}+\nu_{k,2} + 4\nu_{k,3} -2m_{k,3} + 4\nu_{k,4}-2 m_{k,4}),\\
M= \sum_{k=1}^\infty (m_{k,3}+m_{k,4}),\qquad N=\sum_{k=1}^\infty (\nu_{k,1}-\nu_{k,2} + 2\nu_{k,3} -m_{k,3} - 2\nu_{k,4}+m_{k,4} ).
\end{gather*}
This implies that $f_n(z)$ is a linear combination of terms $f_m^{(0)}(z) q^L z^{M}$ with $m+N=n$. We need to prove that, for all these terms, $m+M\leq n+L$ or, equivalently,
\begin{gather*}
L-M+N\geq 0.
\end{gather*}
We insert the constraints on these integers given above to write $L-M+N$ as
\begin{gather*}
\sum_{k=1}^\infty [(k+1)\nu_{k,1} +(k-1)\nu_{k,2} + 2(k+1)\nu_{k,3} + 2k(\nu_{k,3}-m_{k,3}) + 2(k-1)(2\nu_{k,4}-m_{k,4})] ,
\end{gather*}
which is manifestly non-negative.

\subsection[Series representations of $\wp$]{Series representations of $\boldsymbol{\wp}$}\label{appwp}
We derive the formulas in \eqref{wpseries}--\eqref{Snumu}. Recall that $\omega_\nu$ equals $0$, $\pi$, $-\pi-\pi\tau$ and $\pi\tau$ for $\nu=0,1,2$ and $3$, respectively.

We use the absolutely convergent series representation
\begin{gather*}
\frac1{4\sin^2\big(\frac{1}{2} z\big)}=-\sum_{\mu=0}^\infty\mu({\rm e}^{\pm {\rm i} z})^\mu \qquad \mbox{if}\quad \Im(z)\gtrless 0
\end{gather*}
and \eqref{wp22} to compute, assuming $2\pi\Im(\tau)>\Im(x)>0$ (so that $\Im(x+2n\pi\tau)>0$ for $n\geq 0$ and $\Im(x-2n\pi\tau)<0$ for $n>0$),
\begin{gather}
\wp(x|\pi,\pi\tau) + \frac{\eta_1}{\pi} = \sum_{n=0}^\infty \frac1{4\sin^2\big(\frac{1}{2}[x+2n\pi\tau]\big)} + \sum_{n=1}^\infty \frac1{4\sin^2\big(\frac{1}{2}[x-2n\pi\tau]\big)}\nonumber \\ =
-\sum_{n=0}^\infty \sum_{\mu=1}^\infty \mu \big({\rm e}^{{\rm i} (x+2n\pi\tau)}\big)^\mu - \sum_{n=1}^\infty \sum_{\mu=1}^\infty \mu \big({\rm e}^{-{\rm i} (x-2n\pi\tau)}\big)^\mu \nonumber\\ =
 -\sum_{\mu=1}^\infty \mu\left( \sum_{n=0}^\infty {\rm e}^{{\rm i}\mu x} q^{2n\mu} + \sum_{n=1}^\infty {\rm e}^{-{\rm i}\mu x} q^{2n\mu} \right) =
-\sum_{\mu=1}^\infty \mu\left( \frac1{1-q^{2\mu}} {\rm e}^{{\rm i}\mu x} + \frac{q^{2\mu}}{1-q^{2\mu}} {\rm e}^{-{\rm i}\mu x} \right)\label{computation}
\end{gather}
(we inserted \eqref{q} and summed up geometric series; the interchange of summations in the third equality is justif\/ied by absolute convergence).
This implies \eqref{wpseries}--\eqref{Snumu} for $\nu=0$.

From this one gets \eqref{wpseries}--\eqref{Snumu} also in the other cases $\nu=1,2,3$: for $2\pi\Im(\tau)>\Im(x)>0$,
\begin{gather*}
\wp(x+\pi|\pi,\pi\tau) = -\frac{\eta_1}{\pi} -\sum_{\mu=1}^\infty \mu(-1)^\mu \left( \frac1{1-q^{2\mu}} {\rm e}^{{\rm i}\mu x} + \frac{q^{2\mu}}{1-q^{2\mu}} {\rm e}^{-{\rm i}\mu x} \right)
\end{gather*}
implies the result for $\nu=1$; for $\pi\Im(\tau)>\Im(x)>-\pi\Im(\tau)$,
\begin{gather*}
\wp(x+\pi\tau|\pi,\pi\tau) = -\frac{\eta_1}{\pi} -\sum_{\mu=1}^\infty \mu\left( \frac{q^\mu}{1-q^{2\mu}} {\rm e}^{{\rm i}\mu x} + \frac{q^\mu}{1-q^{2\mu}} {\rm e}^{-{\rm i}\mu x} \right)
\end{gather*}
implies the result for $\nu=3$; and for $\pi\Im(\tau)>\Im(x)>-\pi\Im(\tau)$,
\begin{gather*}
\wp(x+\pi+\pi\tau|\pi,\pi\tau) = -\frac{\eta_1}{\pi} -\sum_{\mu=1}^\infty \mu(-1)^\mu \left( \frac{q^\mu}{1-q^{2\mu}} {\rm e}^{{\rm i}\mu x} + \frac{q^\mu}{1-q^{2\mu}} {\rm e}^{-{\rm i}\mu x} \right)
\end{gather*}
and the double periodicity of $\wp$ imply the result for $\nu=2$.

\section{Relation to combinatorics}
\label{appCombinatorics}
In this appendix we give examples that our results allow to translate the problem to solve the Heun equation into combinatorial problems.
We also explain in which sense the representation of the eigenvalues $\mathcal{E}_n$ by Theorem~\ref{Thm2} is simpler than the one provided by Theorem~\ref{Thm1}.

We note that \eqref{Enexplicit} implies the following formula for the Taylor coef\/f\/icients $\mathcal{E}_n^{(\ell)}$ of the gene\-ra\-lized eigenvalues,
\begin{gather}
\mathcal{E}_n^{(\ell\geq 1)} = \sum_{m} \sum_{k_0,k_1,\ldots,k_{m-1}}\sum_{\ell_1,\cdots,\ell_m} \delta \left(\sum_{r=0}^{m-1}k_r, m\right)
\delta\left(\sum_{r=1}^{m-1}rk_r, m-1\right)(m-1)!\nonumber\\
\hphantom{\mathcal{E}_n^{(\ell\geq 1)} =}{}  \times
\delta\left(\ell_1+\cdots+\ell_m ,\ell\right) \frac{\big[\Phi_n^{(r_1,\ell_1)}\big]^{k_1}}{k_1!}\cdots \frac{\big[\Phi_n^{(r_m,\ell_m)}\big]^{k_m}}{k_m!}\label{Enell11}
\end{gather}
with
 \begin{gather}
\Phi^{(r,\ell)}_n \equiv -\sum_{s} \sum_{\mu_1,\ldots,\mu_s}\sum_{r_1,\ldots,r_{s-1}}\sum_{k_1,\ldots,k_s} S_{\mu_1}(k_1)\cdots S_{\mu_s}(k_s) \nonumber\\
\hphantom{\Phi^{(r,\ell)}_n \equiv}{} \times
\frac{\delta(0,\mu_1+\cdots+\mu_s)\delta(\ell,|\mu_1|k_1+\cdots+|\mu_s|k_s) \delta(r_1+\cdots r_{s-1},r)}{\big[b^{(0)}_{n}(\mu_1)\big]^{1+r_1}_n\big[b^{(0)}_{n}(\mu_1+\mu_2)\big]^{1+r_2}_n\cdots \big[b^{(0)}_{n}(\mu_1+\cdots+\mu_{s-1})\big] ^{1+r_{s-1}}_n}, \label{Phin11}
\end{gather}
using notation introduced above (this follows by straightforward computations); for f\/ixed $\ell$, there are only f\/initely many integer vectors $(k_0,\ldots,k_{m-1})$ and $(\ell_1,\ldots,\ell_m)$ that give non-zero contributions to the sums in \eqref{Enell11} (due to the constraints in the Kronecker deltas), and, similarly, \eqref{Phin11} represents each of the f\/initely many dif\/ferent~$\Phi_n^{(r',\ell')}$ needed to evaluate~\eqref{Enell11} as a~f\/inite number of terms. It is clear that this combinatorial problem is complicated.

On the other hand, if one uses \eqref{Enexplicit1} to compute $\mathcal{E}_n^{(\ell)}$, one only needs $\alpha_n^{(\ell')}(n)$ for $\ell'=1,2,\ldots,\ell$, and \eqref{Gnm1} for $m=n$ allows to compute this by a combinatorial formula which is very similar to the one for $\Phi_n^{(0,\ell')}$. It thus is clear that the formula in \eqref{Enexplicit1} gives a signif\/icantly simpler representation of $\mathcal{E}_n^{(\ell)}$ than the formula in \eqref{Enexplicit}, as noted already in Remark~\ref{remComplexity}.

It is interesting to note that, from the combinatorial point of view, the solutions of the eCS model in \cite{ELeCS2} and of the Heun equation in the present paper are very similar; the only dif\/ferences are the details of the building blocks $\Phi^{(r,\ell)}_n$ of the solution. This suggests to us that all elliptic CMS type problems have a common underlying combinatorial structure which deserves further study.

\subsection*{Acknowledgements}
We thank M.~Halln{\"a}s, O.~Chalykh, and H.~Rosengren for helpful discussions and comments, as well as an anonymous referee for carefully reading our paper. We gratefully acknowledge partial f\/inancial support by the \emph{Stiftelse Olle Engkvist Byggm{\"a}stare} (contract 184-0573).

\pdfbookmark[1]{References}{ref}
\LastPageEnding

\end{document}